\documentclass[12pt,draftcls,onecolumn]{IEEEtran}

\usepackage{color, subfigure, graphics, epsfig, amsmath, amssymb, amsbsy, psfrag, bm, url}
\usepackage[justification=centering]{caption}
\usepackage{mdwlist}    
\usepackage{graphicx}
\usepackage{subfigure}

\makeatletter

\def\ps@headings{%
\def\@oddhead{\mbox{}\scriptsize\rightmark \hfil \thepage}%
\def\@evenhead{\scriptsize\thepage \hfil \leftmark\mbox{}}%
\def\@oddfoot{}%
\def\@evenfoot{}}
\makeatother

\makeatletter

\newcommand{\Rmnum}[1]{\expandafter\@slowromancap\romannumeral #1@}
\makeatother

\newtheorem{theorem}{Theorem}
\newtheorem{definition}{Definition}
\newtheorem{lemma}{Lemma}

\newtheorem{corollary}{Corollary}

\newtheorem{remark}{Remark}

\newcommand{\pp}{\mathbf{p}}
\newcommand{\tpp}{\tilde{\pp}}
\newcommand{\QQ}{\mathbf{Q}}

\newcommand{\uu}{\mathbf{u}}

\newcommand{\tc}{\tilde{c}}
\newcommand{\hc}{\hat{c}}
\newcommand{\cc}{\mathbf{c}}
\newcommand{\hcc}{\mathbf{\hc}}
\newcommand{\tcc}{\mathbf{\tc}}
\newcommand{\pqM}{p_{_{qM}}}

\newcommand{\tpqM}{\tilde{p}_{_{qM}}}
\newcommand{\bdelta}{\bm{\delta}}
\newcommand{\bnu}{\bm{\nu}}
\newcommand{\blambda}{\bm{\lambda}}
\newcommand{\btheta}{\bm{\theta}}

\newcommand{\cM}{\mathcal{M}}

\newcommand{\cN}{\mathcal{N}}

\newcommand{\tM}{\tilde{M}}
\newcommand{\tq}{\tilde{q}}

\newcommand{\hM}{\hat{M}}

\newcommand{\nonnegpart}[1]{\left[#1\right]^+}
\newcommand{\hQ}{\hat{Q}}
\newcommand{\tQ}{\tilde{Q}}

\newcommand{\cC}{\mathcal{C}}
\newcommand{\avgQ}{\bar{Q}}

\newcommand{\avgR}{\bar{R}}
\newcommand{\ha}{\hat{a}}

\newcommand{\tm}{\tilde{m}}

\newcommand{\tn}{\tilde{n}}

\newcommand{\tb}{\tilde{b}}

\newcommand{\avgJ}{\bar{J}}

\newcommand{\qq}{\mathbf{q}}
\newcommand{\bmu}{\bm{\mu}}

\newcommand{\hb}{\hat{b}}
\newcommand{\hp}{\hat{p}}
\newcommand{\hpqM}{\hat{p}_{_{qM}}}
\newcommand{\hpp}{\mathbf{\hp}}

\newcommand{\hQQ}{\mathbf{\hQ}}

\newcommand{\mm}{\mathbf{m}}

\begin{document}
\title{Online Advertisement, Optimization\\and Stochastic Networks}

\author{Bo (Rambo) Tan and R. Srikant\\
Department of Electrical and Computer Engineering\\
University of Illinois at Urbana-Champaign\\ 
Urbana, IL, USA}

\maketitle

\begin{abstract}
In this paper, we propose a stochastic model to describe how search service providers charge client companies based on users' queries for the keywords related to these companies' ads by using certain advertisement assignment strategies. We formulate an optimization problem to maximize the long-term average revenue for the service provider under each client's long-term average budget constraint, and design an online algorithm which captures the stochastic properties of users' queries and click-through behaviors. We solve the optimization problem by making connections to scheduling problems in wireless networks, queueing theory and stochastic networks. Unlike prior models, we do not assume that the number of query arrivals is known. Due to the stochastic nature of the arrival process considered here, either temporary ``free''  service, i.e., service above the specified budget (which we call ``overdraft'') or under-utilization of the budget (which we call ``underdraft'') is unavoidable. We prove that our online algorithm can achieve a revenue that is within $O(\epsilon)$ of the optimal revenue while ensuring that the overdraft or underdraft is $O(1/\epsilon)$, where $\epsilon$ can be arbitrarily small. With a view towards practice, we can show that one can always operate strictly under the budget. In addition, we extend our results to a click-through rate maximization model, and also show how our algorithm can be modified to handle non-stationary query arrival processes and clients with short-term contracts.

Our algorithm also allows us to quantify the effect of errors in click-through rate estimation on the achieved revenue. We show that we lose at most $\frac{\Delta}{1+\Delta}$ fraction of the revenue if $\Delta$ is the relative error in click-through rate estimation.

We also show that in the long run, an expected overdraft level of $\Omega(\log(1/\epsilon))$ is unavoidable (a universal lower bound) under any stationary ad assignment algorithm which achieves a long-term average revenue within $O(\epsilon)$ of the offline optimum.
\end{abstract}

\section{Introduction}
\label{sec: intro}
Providing online advertising services has been the major source of revenue for search service providers such as Google, Yahoo and Microsoft. When an Internet user queries a keyword, alongside the search results, the search engine may also display advertisements from some companies which provide services or goods related to this keyword. These companies pay the search service providers for posting their ads with a specified amount of price for each ad 
on a pay-per-impression or pay-per-click basis. We call them ``clients'' 
in the following text.

Maximizing the revenue obtained from their clients is the key objective of search service providers. Research which targets this objective has followed two major directions. One is based on auction theory, in which the goal is to design mechanisms in favour of the service provider, and much of the research in this direction considers static bids (e.g. \cite{IyeKum06_AdAuction}; see \cite{FelMut08_AlgSearchAuc} for a survey), while dynamic models such the one in \cite{MenAsuSri09_DynAdAuction} are still emerging. The other is from the perspective of online resource allocation without considering the impact of the service provider's mechanisms on the clients' bids, and the main focus of this kind of research is on designing an online algorithm which posts specific ads in response to each search query arriving online, in order to achieve a high competitive ratio
with respect to the offline optimal revenue. Our work follows the second direction.

Our model is as follows:

\noindent\hrulefill\\
\textbf{Online Advertising Model:}

Assume that queries for keyword $q$ arrive to the search engine according to a stochastic process at rate $\nu_q$ queries per time slot, where we have assumed that time is discrete and a {\em``time slot''} is our smallest discrete time unit. In response to each query arrival, the search engine may display ads from some clients on the webpage. There are $L$ different places (e.g., top, bottom, left, right, etc.) on a webpage where ads could be displayed. We will call these places {\em``webpage slots."} When client $i$'s ad is displayed in webpage slot $s$ when keyword $q$ is queried, there is a probability with which the user who is viewing the page (the one who generated the query) will click on the ad. This probability, called the {\em``click-through rate,''} is denoted by $c_{qis}.$

A client specifies the amount of money (``bid'') that it is willing to pay to the search service provider when a user clicks on its ad related to a specific query. We use $r_{qi}$ to denote this per-click payment from client $i$ for its ad related to a query for keyword $q$. Additionally, client $i$ also specifies an average budget $b_i$ which is the maximum amount that it is willing to pay per {\em``budgeting cycle''} on average, where a budgeting cycle equals to $N$ time slots (we have introduced the notion of a budgeting cycle since the time-scale over which queries arrive may be different than the time-scales over which budgets may be settled).

The problem faced by the search service provider is then to assign advertisements to webpage slots, in response to each query, so that its long-term average revenue is maximized.

\noindent\hrulefill\\

Based on the above model, we design an online algorithm which achieves a long-term average revenue within $O(\epsilon)$ of the offline optimal revenue, where $\epsilon$ can be chosen arbitrarily small, indicating the near-optimality of our online algorithm. Before entering into the details, in the next two subsections we will first survey the related literature, highlight the main contributions of our work, and discuss the differences between our model and previous ones.

\subsection{Related Work}
\label{sec: survey}
We will only survey the online resource allocation models here, and not the auction models. The online ads model in prior literature mainly include two types, namely AdWords (AW) and Display Ads (DA), of which the difference lies in the constrained resource of each client. In the AW model, the resource is the client's budget, while in the DA model, the resource is the maximum number of impressions agreed on by the client and the service provider. Correspondingly, after each resource allocation step, the resource of a client whose ad is posted, is reduced by the bid value\footnote{This refers to the pay-per-impression scheme. With a pay-per-click scheme, the reduction only happens if the ad is clicked.} in the AW model, or $1$ impression in the DA model. Both of them belong to a general class of packing linear programs formulated in \cite{FelHenKor10_ads}. Most of the prior online algorithms for solving the AW and DA model respect the hard constraint on the client's resources. One exception is \cite{FelHenKor09_adsFreeDisposal}, where the authors argue that ``free disposal'' of resources makes the DA model more tractable (but not necessary for the AW model). 

Mehta et al. \cite{mehsabume_FOCS05} modeled the online ads problem as a generalization of an online matching problem \cite{KarpVaz_STOC90_matching} on a bipartite graph of queries and clients. Later in \cite{BucJaiJos07_primal-dual}, Buchbinder et al. showed that matching clients to webpage slots (whether it is a single slot or multiple slots) can be solved as a maximum-weighted matching problem. Following \cite{BucJaiJos07_primal-dual}, a number of other online algorithms using the maximum-weighted bipartite matching idea have been proposed in \cite{MahNazSab07_AdEstimate, FelHenKor09_adsFreeDisposal, DevHay09_ads} and \cite{FelHenKor10_ads}.
The algorithms in \cite{KalPruhs00_b-matching} and \cite {mehsabume_FOCS05}, which were earlier than \cite{BucJaiJos07_primal-dual}, can also be regarded as maximum-weighted matching solutions on this bipartite graph of clients and webpage slots.

In \cite{KalPruhs00_b-matching}, the ``b-matching'' problem (related to the online ads context, bids are trivially $0$ or $1$ and budgets are all $b$) is solved by an $1-1/e$ competitive algorithm as $b\rightarrow \infty$ and the weights are the remaining budgets of those clients interested in the newly arrived query (i.e., the bid equals $1$). For the online ads problem in which bids and budgets can have general and different values, \cite{mehsabume_FOCS05} (its longer version is \cite{mehsabume_JACM07}) uses the ``discounted'' bids as the weights corresponding to each client. The discount factor is calculated by a function $\psi(x) = 1-e^{x-1}$, of which the input $x$ is the fraction of a client's budget that has been consumed. Their algorithm is also $1-1/e$ competitive, under an assumption that bids are small compared to budgets. By taking advantage of estimated numbers of query arrivals for each keyword within a given period and modifying the discount factor in \cite{mehsabume_FOCS05}, Mahdian et al. \cite{MahNazSab07_AdEstimate} designed a class of algorithms which achieve a considerably better competitive ratio with accurate estimates while still guarantee a reasonably good competitive ratio with inaccurate estimates, also assuming small bids.

The algorithms in \cite{BucJaiJos07_primal-dual, FelHenKor09_adsFreeDisposal, DevHay09_ads,  FelHenKor10_ads} and \cite{AgrawalWangYe_TR10}, all use a primal-dual framework to compute a maximum-weighted matching at each iteration, in which the dual variables (corresponding to each client) are used to determine the weights. The two $1-1/e$ competitive 
 algorithms in \cite{BucJaiJos07_primal-dual} and \cite{FelHenKor09_adsFreeDisposal} update the dual variables dynamically 
in their primal-dual type algorithms every time a decision is made. Specifically, 
each dual variable in \cite{BucJaiJos07_primal-dual}, which implicitly tracks the fraction of budget that has been spent by the corresponding client, grows during each iteration at a rate parameterized by the fraction of the bid for the incoming query in this client's total budget, while \cite{FelHenKor09_adsFreeDisposal} uses an ``exponentially weighted average'' of the up-to-date $n(i)$ most valuable impressions\footnote{In the DA model in \cite{FelHenKor09_adsFreeDisposal}, $n(i)$ is defined as the maximum number of impressions agreed for client $i$. After allowing free disposal, only the current $n(i)$ most valuable impressions assigned to client $i$ will be considered.} assigned to client $i$ as a new dual variable with respect to this client. On the other hand, the three dual type learning-based algorithms in \cite{DevHay09_ads}, \cite{FelHenKor10_ads} and \cite{AgrawalWangYe_TR10} achieve a competitive ratio of $1-O(\epsilon)$ based on a random-order arrival model (rather than the adversarial model in most of the earlier work), 
assuming small bids and knowledge of the total number of queries. The main difference between them is that \cite{DevHay09_ads} and \cite{FelHenKor10_ads} use an initial $\epsilon$ fraction of queries to learn the optimal dual variables (with respect to this training set), 
while the algorithm in \cite{AgrawalWangYe_TR10} repeats the learning process over geometrically growing intervals. Additionally, the ``small bids'' condition in \cite{AgrawalWangYe_TR10} is slightly weaker than the condition in  \cite{DevHay09_ads} and \cite{FelHenKor10_ads}.

\subsection{Our Contributions and Comparison to Prior Work}
\label{sec: contribution}
As in prior work (especially \cite{BucJaiJos07_primal-dual} and \cite{FelHenKor09_adsFreeDisposal}), our solution relies on a primal-dual framework to solve a maximum-weighted matching problem on a bipartite graph of clients and webpage slots, with dynamically updated dual variables which contribute to the weights on the edges of the bipartite graph. However, unlike prior work, we are able to obtain a revenue which is $O(\epsilon)$ close to the optimal revenue using a purely adaptive algorithm without the need for the knowledge of the number of query arrivals over a time period or the average arrival rates.

Our solution is related to scheduling problems in wireless networks. In particular, we use the optimization decomposition ideas in  \cite{GeoNeeTas06_ResAllocWireless}, the stochastic performance bounds in \cite{LinShrSri06_OptimizeWireless} and the modeling of delay-sensitive flows in \cite{juasri10}. Borrowing from that literature, we introduce the concept of an {\em``overdraft'' queue}. The overdraft queue measures the amount by which the provided service temporarily exceeds the budget specified by a client. In making the connection to wireless networks, we define something called the {\em``per-client revenue region,''} which is related to the concept of capacity region in queueing networks (see \cite{GeoNeeTas06_ResAllocWireless, LinShrSri06_OptimizeWireless}). In our context, it characterizes the revenue extractable from each client as a function of all the clients' budgets.

Our online algorithm exhibits a trade-off between the revenue obtained by the service provider and the level of  overdrafts. We can further modify our online algorithm so that clients can always operate strictly under their budgets. Finally, our algorithm and analysis naturally allow us to assess the impact of click-through rate estimation on the service providers revenue.

We are able to show that our online algorithm achieves an overdraft level of $O(1/\epsilon)$. So a natural question is whether this bound is tight. We show that the overdraft for any algorithm must be $\Omega(\log(1/\epsilon))$. While there is a gap between the upper and lower bounds, together they imply that the overdraft must increase when $\epsilon$ goes to zero. This work is related to \cite{BerGal_ToIT02_EnergyDelayTradeoff, Neely_TAC09_EnergyDelayTradeoff, shah_Allerton08_LB,neely2007optimal} and \cite{huang11_LIFO} in the context of communication networks. See Section \ref{sec: que_len_lb} for a detailed survey.

Besides the revenue maximization model, we also study another online ads model in which the objective is to maximize the average overall click-through rate, subject to a minimum impression requirement for each client. We also show that our results can be naturally extended to handle non-stationary query arrival processes and clients which have short-term contracts with the service provider.
.

Like the algorithm in \cite{AgrawalWangYe_TR10}, our algorithm can also be generalized to a wider class of linear programs within different application contexts, where the coefficients in the objective function and constraints are not necessarily nonnegative.

There are two points of departure in our algorithm compared to existing models: the first one is that we assume a purely stochastic model in which the query arrival rates are unknown. Thus, there is no need to know the number of arrivals in a time period as in prior models, and this is even true for non-stationary query arrival processes. The other is that we assume an average budget rather a fixed budget over a time horizon. This allows us to better model permanent clients (e.g., big companies who do not stop advertising) and who do not provide a fixed time-horizon budget. Clients who advertise for a limited amount of time can also be handled well since the algorithm is naturally adaptive. 

A minor difference with respect to prior models is that our model assumes that time is slotted. This can be easily modified to assume that query arrivals can occur at any time according to some continuous-time stochastic process. The only difference is that our analysis would then involve continuous-time Lyapunov drift instead of the discrete-time drift used in this paper. From a theoretical point of view, our analysis is different from prior work which uses competitive ratios: our model and solution is similar in spirit to stochastic approximation \cite{borkar08stochastic} where gradients (here the gradient of the dual objective) are known only with stochastic perturbations. This point of view is essential to model stochastic traffic with unknown statistics.

Instead of the $1-O(\epsilon)$ competitive ratio in prior work, we show that our algorithm achieves a revenue which is within $O(\epsilon)$ of the optimal revenue. The $O(\epsilon)$ penalty arises due to the stochastic nature of our model. However, we do not require assumptions such as 
knowledge of the total number of queries in a given period  \cite{MahNazSab07_AdEstimate, DevHay09_ads, FelHenKor10_ads,AgrawalWangYe_TR10}, or information of keyword frequencies \cite{MahNazSab07_AdEstimate}.\footnote{It should be mentioned that another common assumption ``small bids'' (or ``large budgets'', ``large offline optimal value'') used in \cite{KalPruhs00_b-matching, mehsabume_FOCS05, MahNazSab07_AdEstimate, FelHenKor09_adsFreeDisposal,DevHay09_ads} and \cite{FelHenKor10_ads} is not essentially different from our ``long-term'' assumption.}

\subsection{Organization of the Paper}
The rest of the paper is organized as follows: In Section~\ref{sec: static_optimize}, we formulate an optimization problem involving long-term averages. In Section \ref{sec: dyn_alg}, we start considering the stochastic version of our model and propose an online algorithm, which also introduces the concept of ``overdraft queue.'' Performance analysis of this online algorithm, which includes the near-optimality of the long-term revenue and an upper bound on the overdraft level, will also be done in Section \ref{sec: dyn_alg}. The last two subsections of Section~\ref{sec: dyn_alg} 
present two extensions, namely the decisions based on estimated click-through rates and the ``underdraft'' mechanism. In Section \ref{sec: que_len_lb},  we derive a universal lower bound on the expected overdraft level under any stationary algorithms for online advertising. 
The second online ads model ``click-through rate maximization problem'' with its related extensions, algorithm design and analysis is given in Section \ref{sec: click-through}. 
Section~\ref{sec: conclusion} concludes the whole paper.

Compared to an earlier version of this paper which appeared in \cite{tansri_adwords_CDC11}, we give a more detailed literature survey in Subsection \ref{sec: survey}, all the proofs for the lemmas, theorems and corollaries in Section \ref{sec: dyn_alg} (we only stated these results without proofs in \cite{tansri_adwords_CDC11} due to page limits), and full discussions on the underdraft mechanism in Subsection \ref{sec: negative_queue}. Sections \ref{sec: que_len_lb} and \ref{sec: click-through} are completely new.

\section{An Optimization Problem Involving Long-Term Averages}
\label{sec: static_optimize}
Based on the model described in Section \ref{sec: intro}, we first pose the revenue maximization problem as an optimization problem involving long-term averages. For this purpose, we define an assignment of clients to webpage slots as a matrix $M$ of which the $(i,s)^{\rm th}$ element is defined as follows:
$$M_{is}=\left\{
\begin{array}{ll}
1,&\mbox{if client } i \mbox{ is assigned to webpage slot } s\\
0,&else.
\end{array}
\right.
$$
The matrix $M$ has to satisfy some practical constraints. First, a webpage slot can be assigned to only one client and vise versa.
Furthermore, the assignment of clients to certain webpage slots may be prohibited for certain queries. For example, it may not make sense to advertise chocolates when someone is searching for information about treatments for diabetes. These constraints can be abstracted as follows:
For the queried keyword $q,$ the set of assignment matrices have to belong to some set $\cM_q.$ We also let $\pqM$ be the probability of choosing matrix $M$ when the queried keyword is $q.$

The optimization problem is then given by
\begin{equation}\label{eq: objective}
\max_{\pp} \avgR(\pp) = \sum_q\nu_q  \sum_{M\in \cM_q}\pqM \sum_{i,s} M_{is}c_{qis}r_{qi}
\end{equation}
subject to
\begin{eqnarray}
N \sum_q\nu_q  \sum_{M\in \cM_q}\pqM \sum_{s} M_{is}c_{qis}r_{qi}\leq b_i,\qquad \forall i;\label{eq: budget}\\
0\leq \pqM \leq 1,\qquad \forall q,~M \in \cM_q;\label{eq: prob}\\
\sum_{M\in \cM_q} \pqM \leq 1, \qquad \forall q.\label{eq: prob_sum}
\end{eqnarray}
In the above formulation, the objective (\ref{eq: objective}) is the average revenue per time slot and constraint (\ref{eq: budget}) expresses the fact that the average payment over a budgeting cycle should not exceed the average budget. The optimization is a linear program and if all the problem parameters are known, in principle, it can be solved offline, returning probabilities $\{\pqM\}$ which can be used by a service provider to maximize its revenue. However, such an offline solution is not desirable for at least two reasons:
\begin{itemize*}
\item Being a static approach, it does not use any feedback about the current state of the system. For example, the fact that the empirical average payment of a client has severely exceeded its average budget would have no impact on the subsequent assignment strategy. Since the formulation and hence, the solution, only cares about long-term budget constraint satisfaction, severe overdraft or underdraft of the budget can occur over long periods of time.

\item The offline solution is a function of the query arrival rates $\{\nu_q\}.$ Thus, a change in the arrival rates would require a recomputation of the solution.
\end{itemize*}

In view of these limitations of the offline solution, we propose an online solution which adaptively assigns client advertisements to webpage slots to maximize the revenue. As we will see, the online solution does use feedback about the overdraft (or underdraft) level in future decisions, and does not require knowledge of $\{\nu_q\}$.

\section{Online Algorithm and Performance Analysis}
\label{sec: dyn_alg}
\subsection{A Dual Gradient Descent Solution}
\label{sec: iterative_sol}
To get some insight into a possible adaptive solution to the problem, we first perform a dual decomposition which suggests a gradient solution. However, a direct gradient solution will not take into the account the stochastic nature of the problem and will also require knowledge of the query arrival rates $\{\nu_q\}.$ We will address these issues in the following subsections, using techniques that, to the best of our knowledge, have not been used in prior literature on the online advertising problem.

We append the constraint (\ref{eq: budget}) to the objective (\ref{eq: objective}) using Lagrange multipliers $\delta_i \geq 0$ to obtain a partial Lagrangian function
\begin{eqnarray*}
L(\pp,\bdelta)\!\!\!\!&=&\!\!\!\!\sum_q\nu_q  \sum_{M\in \cM_q} \pqM \sum_{i,s} M_{is}c_{qis}r_{qi} -\sum_i \delta_i \cdot 
\left(\sum_q\nu_q  \sum_{M\in \cM_q}\pqM\sum_{s} M_{is}c_{qis}r_i-\frac{b_i}{N}\right) \nonumber\\
&=&\!\!\!\!\sum_q\nu_q \!\! \sum_{M\in \cM_q} \!\!\!\!\pqM \!\!\sum_{i,s} M_{is}c_{qis}r_{qi}(1-\delta_i)+\!\!\sum_i \frac{\delta_i b_i}{N},\label{eq: Lagrangian}
\end{eqnarray*}
subject to constraints (\ref{eq: prob}) and (\ref{eq: prob_sum}). The dual function is
\begin{equation*}
D(\bdelta) 
= \max_{\pp} \sum_q\nu_q \!\!\!\! \sum_{M\in \cM_q}\!\!\pqM\!\! \sum_{i,s} M_{is}c_{qis}r_{qi}(1 - \delta_i) +  \sum_{i}\frac{\delta_i b_i}{N}, \label{eq: dual}
\end{equation*}
subject to constraints (\ref{eq: prob}) and (\ref{eq: prob_sum}). Note that the maximization part in the dual function can be decomposed into independent maximization problems with regard to each queried keyword $q$, i.e., for all $q$,
\begin{eqnarray*}
&&\max_{\{\pqM,~M\in \cM_q\}} \sum_{M\in\cM_q} \pqM \sum_{i,s}M_{is}c_{qis}r_{qi}(1-\delta_i)
= \max_{M\in\cM_q} \sum_{i,s} M_{is}c_{qis}r_{qi}(1-\delta_i), \label{eq: max_each_query}
\end{eqnarray*}
where it is easy to see that each maximization is solved by a deterministic solution. This suggests the following primal-dual algorithm to iteratively solve the original optimization problem (\ref{eq: objective}): at step $k$,
\begin{eqnarray*}
\forall q,&&\hM^*(q,k) \in \arg \max_{M\in \cM_q} \sum_{i,s} M_{is} c_{qis}r_{qi}(1-\delta_i(k)); \label{eq: primal_dual_max}\\
\forall i,&&\delta_i(k+1) = \Bigg[\delta_i(k)+\epsilon\bigg(
  N \sum_q\nu_q \sum_{s} [\hM^*(q,k)]_{is}\cdot 
  c_{qis}r_{qi} - b_i\bigg)\Bigg]^+, \label{eq: primal_dual_update_Lagrangian}
\end{eqnarray*}
where $\epsilon>0$ is a fixed step-size parameter, and $[x]^+ = x~\mbox{if}~x\ge 0$ or $[x]^+ = 0$ otherwise. Furthermore, defining $\hQ_i (k)\triangleq \delta_i(k) / \epsilon$, the above iterative algorithm becomes
\begin{eqnarray*}
\forall q,&&\hM^*(q,k) \in \arg \max_{M\in \cM_q} \sum_{i,s} M_{is}c_{qis}r_{qi}\left(\frac{1}{\epsilon}-\hQ_i(k)\right); \label{eq: primal_dual_max2}\\
\forall i,&&\hQ_i(k+1) = \nonnegpart{\hQ_i(k)+ \hat{\lambda}_i(k)
     - b_i}, \label{eq: primal_dual_update_avg_queue}
\end{eqnarray*}
where
\begin{equation}
\hat{\lambda}_i(k) \triangleq N \sum_q\nu_q \sum_{s} [\hM^*(q,k)]_{is}c_{qis}r_{qi}. \label{eq: def_lambda_hat}
\end{equation}
Note that $\hQ_i(k)$ can be interpreted as a queue which has $\hat{\lambda}_i(k)$ arrivals and $b_i$ departures at step $k$. Although this algorithm already uses the feedback provided by $\{\mathbf{\hQ}(k)\}$ (or $\{\bm{\delta}(k)\}$) about the state of the system, it is still using a priori information about the arrival rates of queries in $\{\bm{\hat{\lambda}}(k)\}$, hence not really ``online.'' However, it motivates us to incorporate a queueing system with stochastic arrivals into the real online algorithm, which will be described in the next subsection.

\subsection{Stochastic Model, Online Algorithm, and ``Overdraft Queue''}
\label{sec: stochastic}
In practice, a search service provider may not have a priori information about the query arrival rates $\{\nu_q\}$, and generally, query arrivals during each time slot are stochastic rather than constant. Let time slots be indexed by $t \in \mathcal{Z}^+ \cup \{0\}$. We specify our detailed statistical assumptions as follows:
\begin{itemize}
\item Query arrivals: Assume that a time slot is short enough so that query arrivals in each time slot can be modeled as a Bernoulli random variable with occurrence probability $\nu$. The probability that an arrived query is for keyword $q$ is assumed to be $\vartheta_q$ and $\sum_q \vartheta_q = 1$. 
Let $\tq(t)$ represent the index of the keyword queried in time slot $t$, such that $\tq(t) = q$ w.p. $\nu_q = \nu \vartheta_q$ for all $q$ (indexed by positive integers) 
and $\tq(t) = 0$ w.p. $1-\nu$, which accounts for the case that no query arrives.

\item Budget spending: We limit the values of budget spent in each budgeting cycle to be integers. To match the average budget $b_i$ (when it is not an integer), the budget of client $i$ in budgeting cycle $k$ is assumed to be a random variable $\tb(k)$ which equals $\lceil b_i \rceil$ w.p. $\varrho_i$ and $\lfloor b_i \rfloor$ otherwise, such that
$E[\tb(k)] = \varrho_i \lceil b_i \rceil + (1-\varrho_i) \lfloor b_i \rfloor = b_i$, i.e., $\varrho_i = \frac{b_i- \lfloor b_i \rfloor}{\lceil b_i \rceil- \lfloor b_i \rfloor} = b_i - \lfloor b_i \rfloor.$ For the trivial case that $b_i$ is already an integer, we let $\varrho_i = 1$.

\item Click-through behaviors: In time slot $t$, after a query for keyword $q$ arrives, if the ad of client $i$ is posted on webpage slot $s$ in response to this query, then whether this ad will be clicked is modeled as a Bernoulli random variable  $\tc_{qis}(t)$ with occurrence probability $c_{qis}$.
\end{itemize}

We now want to implement the above iterative algorithm online based on this stochastic model. According to  definition~(\ref{eq: def_lambda_hat}), $\hat{\lambda}_i$ includes average query arrivals and click-through choices within $N$ time slots (i.e., one budgeting cycle). Thus,  each iteration step in the online algorithm should correspond to a budgeting cycle. For convenience, we define
$$
\uu(k) \triangleq \{\tq(t),\tcc(t)~\mbox{for}~kN \le t \le kN+N-1\}
$$
as a collection of random variables describing user behaviors (including stochastic query arrivals and click-through choices) in budgeting cycle $k$.
The online algorithm is then described as follows:

\noindent\hrulefill\\
\textbf{Online Algorithm:} (in each budgeting cycle $k\ge 0$)
\\

In each time slot $t \in [kN, kN+N-1]$, if $\tq(t) > 0$, choose the assignment matrix
\begin{eqnarray}
&&\tM^*(t,\tq(t),\QQ(k)) 
\in \arg \max_{M\in \cM_{\tq(t)}} \sum_{i,s} M_{is} c_{\tq(t) is}r_{\tq(t) i}\left(\frac{1}{\epsilon}-Q_i(k)\right). \label{eq: dyn_alg_choose_mat}
\end{eqnarray}

At the end of budgeting cycle $k$, for each client $i$, update
\begin{equation}
Q_i(k+1) = \nonnegpart{Q_i(k)+ A_i(k,\QQ(k),\uu(k)) - \tb_i(k)}, \label{eq: dyn_alg_que_update}
\end{equation}
where
\begin{eqnarray}
\!\!\!\!\!\!&&A_i(k,\QQ(k),\uu(k)) 
\triangleq \sum_{t = kN}^{kN+N-1} \!\!\!\! \sum_{s} [\tM^*(t, \tq(t), \QQ(k))]_{is}\cdot \tc_{\tq(t)is}(t) \cdot r_{\tq(t)i}. \label{eq: virtual_arr}
\end{eqnarray}


\noindent\hrulefill\\
Here, $A_i(k,\QQ(k),\uu(k))$ represents the revenue obtained by the service provider from client $i$ during budgeting cycle $k$, and recall that $\tilde{b}_i(k)$ is a random variable which takes integer values whose mean is equal to the average budget per budgeting cycle.

In this algorithm, client $i$ is associated with a virtual queue $Q_i$ (maintained at the search service provider). During budgeting cycle $k$, the amount of money client $i$ is charged by the search service provider $A_i(k,\QQ(k),\uu(k))$ is the arrival to this queue, and the average budget per budgeting cycle $b_i$ is the departure from this queue. Note that if this queue is positive, it means that the total value of the real service already provided to the client has temporarily exceeded the client's budget, i.e., ``free'' service has been provided temporarily.
Hence, we call this queue the ``overdraft queue.''

There are two different time scales here. The faster one is a time slot, the smallest time unit used to capture user behaviors (including stochastic query arrivals and click-through choices) and execute ad-posting strategies. The slower one is a budgeting cycle (equal to $N$ time slots), at the end of which the overdraft queues are updated based on the revenue obtained over the whole budgeting cycle.

We make the following assumptions on the above stochastic model:
$\{\tq(t)\}$ are i.i.d. across time slots $t$;
$\{\tc_{qis}(t)\}$ are independent across $q$, $i$, $s$, and $t$; each variable in $\{\tq(t)\}$ and each variable in $\{\tc_{qis}(t)\}$ are mutually independent. In fact, the model can be generalized to allow for query arrivals correlated over time and across keywords, and other similar correlations inside the click-through choices or between these two stochastic processes.
Such models would only make the stochastic analysis more cumbersome, but the main results will continue to hold under these more general models.

In order to guarantee that the Markov chain which we will define later is both irreducible and aperiodic, we further assume that the probability of whether there is an arrival in a time slot $\nu \in (0,1)$. 
We also assume that $r_{qi}$ for all $q$ and $i$ can only take integer values. Together with the fact that $\tb(k)$ takes integer values, $\{\QQ(k)\}$ becomes a discrete-time integer-valued queue. Note that assuming integer values 
is only for ease of analysis, but not necessary.

\subsection{An Upper Bound on the Overdraft}
\label{sec: que_len_ub}
%
%

According to the ad assignment step (\ref{eq: dyn_alg_choose_mat}), if at the beginning of budgeting cycle $k$, $Q_i(k) > 1/\epsilon$, then for this budgeting cycle, the $i^{\rm th}$ row of $\tM^*(t,q,\QQ(k))$ is always a zero vector, i.e., the service provider will not post the ads of client $i$ until $Q_i(k)$ falls below $1/\epsilon$. Since by assumption the number of query arrivals per time slot is upper bounded, for any budgeting cycle $k$, one can bound the transient length of each overdraft queue as below:
$$
Q_i(k)  \le \frac{1}{\epsilon}+ N \cdot \arg\max_{q,s}\{r_{qi} c_{qis}\} - \lfloor b_i \rfloor,~\forall i.
$$
Therefore, $Q_i(k) \sim O(1/\epsilon)$ for all $i$, and stability is not an issue for these ``upper bounded'' queues. It further implies that this online algorithm satisfies the budget constraints in the long run, i.e., for all client $i$,
\begin{equation}
\lim_{K\rightarrow\infty} E\left[\frac{1}{K}\sum_{k=0}^{K-1} A_i(k,\QQ(k),\uu(k))\right] \le b_i \label{eq: long-run spending}
\end{equation}
must hold.

It should be mentioned that in \cite{huang11_LIFO}, through using the LIFO queueing discipline, the authors show an $O((\log(1/\epsilon))^2)$ bound on the averaged waiting time encountered by most of the packets, which is tighter than 
the bound $O(1/\epsilon)$ under the FIFO queueing discipline (see e.g. \cite{GeoNeeTas06_ResAllocWireless}; our above result also fits this bound). While the length of a FIFO queue is proportional to the arrival rate according to Little's law \cite{asmussen03AP_que}, the length of a LIFO queue in \cite{huang11_LIFO} is still $O(1/\epsilon)$, even if it is occupied by very ``old'' packets which only accounts for a negligible fraction $O(\epsilon^{\log(1/\epsilon)})$ of all the packets that have arrived. Unlike in a  communication network where waiting time is usually the main concern and dropping a small fraction of old packets does almost no hurt to many online applications, what clients of online advertising service care about is how much they have paid beyond their budgets, which is measured by the overdraft queue in our model. 


%

\subsection{Near-Optimality of the Online Algorithm}
\label{sec: converge}
We now show that, in the long term, the proposed online algorithm achieves a revenue that is close to the optimal revenue $\avgR(\pp^*)$ (where $\pp^*$ is the solution to the optimization problem \eqref{eq: objective}). 
We start with the following lemma:\\

\begin{lemma}
\label{lemma: bounded_drift2}
Consider the Lyapunov function $V(\QQ) = \frac{1}{2} \sum_{i} Q_i^2$. For any $\epsilon>0$, and each time period $k$,
\begin{eqnarray*}
&&E[V(\QQ(k+1))|\QQ(k)=\QQ] - V(\QQ)
\le - \frac{N}{\epsilon} \left(\bar{R}(\pp^*) - \bar{R}(\tpp^*(k,\QQ)) \right) + B_1
    - B_2\sum_i Q_i.
\end{eqnarray*}
Here,
\begin{eqnarray}
\!\!\!\!\!\!\!\!B_1 \!\!\!\!&\triangleq&\!\!\!\! \frac{1}{2}\Big(
(N(N-1)L^2+N\!L) (\arg\max_{q,i,s}\{c_{qis} r_{qi}\})^2\nonumber\\
&&\!\!\!\!+ \sum_i \lceil b_i \rceil^2 (b_i- \lfloor b_i \rfloor)
+ \lfloor b_i \rfloor^2 (1-b_i+\lfloor b_i \rfloor)
\Big),\label{eq: B1_def}
\end{eqnarray}
where $L$ is the number of webpage slots;
\begin{eqnarray}  
\!\!\!\!\!\!\!\!B_2 \!\!\!\!&\triangleq&\!\!\!\! \min_{i} \{b_i \!-\! N\sum_q \nu_q \!\! \sum_{M\in\cM_q} \pqM^* \sum_{s} M_{is}c_{qis}r_{qi}\}; \label{eq: B2_def}
\end{eqnarray}
and $\tpp^*(k,\QQ) \triangleq \{\tpqM^*(k,\QQ),~\forall q,M\in \cM_q\}$ where $\tpqM^*(k,\QQ)$ equals $1$ if $M=\tM^*(t,q,\QQ)$  for $kN\le t \le kN+N-1$ (i.e., the optimal matrix in the maximization step (\ref{eq: dyn_alg_choose_mat})) and $0$ otherwise.
\hfill $\diamond$
\end{lemma}


The proof is given in Appendix \ref{sec: proof_bounded_drift2}.\\


Now we are ready to present one of the major theorems in this paper, indicating that the long-term average revenue achieved by our online algorithm is within $O(\epsilon)$ of the maximum revenue obtained by the offline optimal solution. The proof is given in Appendix~\ref{sec: proof_revenue_converge}.
\\

\begin{theorem}
\label{thm: revenue_converge}
For any $\epsilon >0$,
$$
0 \le \lim_{K\rightarrow \infty} E\left[ \bar{R}(\pp^*) - \frac{1}{KN}\sum_{k=0}^{K-1} R(k) \right] \le \frac{B_1\epsilon}{N}
$$
for some constant $B_1 > 0$ (defined in \eqref{eq: B1_def} in Lemma~\ref{lemma: bounded_drift2}), where $R(k) \triangleq \sum_i A_i(k,\QQ(k),\uu(k)).$ is defined as  
the revenue obtained during budgeting cycle $k$. \hfill $\diamond$
\end{theorem}
\

\begin{remark}
If we choose a very small $\epsilon$, the matching in \eqref{eq: dyn_alg_choose_mat} behaves like a greedy solution until the queue lengths grows comparably large. This indicates a tradeoff between how close to the long-term optimal revenue the algorithm can achieve and the actual convergence time. 

Additionally, supposing that $\{r_{qi}\}$ and $\{b_i\}$ are both measured in another scale with a factor $\alpha$, e.g., using cents instead of dollars ($\alpha = 100$),  and assuming that $\alpha$ is unknown, it can be shown that the $O(\epsilon)$ convergence bound will also be scaled by $\alpha$ if we measure the revenue in the original scale. To change the algorithm into a ``scale-free'' version, $\{r_{qi}\}$ and $\{b_i\}$ should be divided by a common benchmark value, e.g., the largest budget specified by all the initially existing clients. 
Since the benchmark value is also implicitly multiplied by $\alpha$ if measured in another scale, the scaling factor will be canceled in the normalized $\{r_{qi}\}$ and $\{b_i\}$ and no longer affect the convergence bound. \hfill $\diamond$
\end{remark}


\subsection{Impact of Click-Through Rate Estimation}
\label{sec: click-through_estimate}
In our online algorithm, the decision of picking an optimal ad assignment matrix in (\ref{eq: dyn_alg_choose_mat}) in response to each query is based on the true click-through rates $\cc$. In reality, an estimate $\hcc$ based on historical click-through behaviors is used, i.e., in response to each query for keyword $q$, which arrives in time slot $t \in [kN, kN+N-1]$, we choose the assignment matrix
\begin{eqnarray}
&&\tM^*(t,\tq(t),\QQ(k)) 
\in \arg \!\! \max_{M\in \cM_{\tq(t)}}\! \sum_{i,s} M_{is} \hc_{\tq(t) is}r_{\tq(t) i}\!\left(\frac{1}{\epsilon}-Q_i(k)\right)\!. 
\label{eq: dyn_alg_choose_mat_estimate}
\end{eqnarray}
We then have the following corollary in addition to Theorem~\ref{thm: revenue_converge} in Subsection~\ref{sec: converge}:\\

\begin{corollary}
\label{cor: revenue_converge_estimate}
Assume that the estimated click-through rates $\hcc \in [\cc(1-\Delta), \cc(1+\Delta)]$ with some $\Delta \in (0,1)$. Under our online algorithm with estimated click-through rates, $\QQ(k)$ is still positive recurrent. 
Then, for any $\epsilon >0$,
$$
\lim_{K\rightarrow \infty} E\left[\frac{1}{KN}\sum_{k=0}^{K-1} R(k) \right] \ge
\left(\frac{1-\Delta}{1+\Delta}\right)\cdot \bar{R}(\pp^*) - \frac{B_1\epsilon}{N},
$$
for some constant $B_1 > 0$ (defined in equation~(\ref{eq: B1_def}) in Lemma~\ref{lemma: bounded_drift2}).  \hfill $\diamond$
\end{corollary}

Proving this needs some minor changes to the proof of Lemma \ref{lemma: bounded_drift2} and Theorem \ref{thm: revenue_converge}, which will be shown in Appendix~\ref{sec: proof_revenue_converge_estimate}.

\

\begin{remark}
Corollary~\ref{cor: revenue_converge_estimate} tells us that for small $\epsilon$, the long-term average revenue achieved by our online algorithm with estimated click-through rates will be at least $\left(\frac{1-\Delta}{1+\Delta}\right)$ of the offline optimal revenue.
\hfill $\diamond$
\end{remark}

\subsection{Underdraft: Staying under the Budget}
\label{sec: negative_queue}
In the previous sections, we allowed the provision of temporary free service to clients, which we call overdraft. If this is not desirable for some reason, the algorithm can be modified to have non-positive overdraft. 
We do this by allowing the queue lengths to become negative, but not positive. The practical meaning of negative queue lengths is to allow each client to accumulate a certain volume of ``credits'' if the current budget is under-utilized and use these credits to offset future possible overdrafts. We call this negative queue length \textit{``underdraft.''} 
Corresponding to this mechanism, we modify our online algorithm
as follows: in response to each query for keyword $q$, which arrives in time slot $t \in [kN, kN+N-1]$, choose the assignment matrix
\begin{eqnarray*}
&&\tM^*(t,\tq(t),\QQ(k)) 
\in \arg\!\! \max_{M\in \cM_{\tq(t)}} \!\!\sum_{i,s} M_{is} c_{\tq(t) is}r_{\tq(t) i}\left(\Gamma_i-Q_i(k)\right), \label{eq: dyn_alg_choose_mat2}
\end{eqnarray*}
and at the end of budgeting cycle $k$, for each client $i$, update
\begin{equation*}
Q_i(k+1) = \max\{Q_i(k)+ A_i(k,\QQ(k),\uu(k)) - \tb_i(k),-C_i\}, \label{eq: dyn_alg_que_update2}
\end{equation*}
where $\Gamma_i$ denotes a customized ``throttling threshold'' (not necessarily $1/\epsilon$) and $C_i$ denotes the maximum allowable credit volume for client $i$. Recall that $A_i(k,\QQ(k),\uu(k))$ is defined in equation~(\ref{eq: virtual_arr}).

We can bound each overdraft queue as below:
$$
Q_i(k)  \le \Gamma_i + N \cdot \arg\max_{q,s}\{r_{qi} c_{qis}\} - \lfloor b_i \rfloor,~\forall i,k.
$$
Thus, if our objective is to eliminate overdrafts (i.e., $Q_i(k) \le 0$ for all $k$), we can set
\begin{equation}
\Gamma_i := \left[\lfloor b_i \rfloor-N \cdot \arg\max_{q,s}\{r_{qi} c_{qis}\} \right]^-,~\forall i, \label{eq: def_Ui}
\end{equation}
where in contrary to $[x]^+$, $[x]^-$ takes the non-positive part of $x$, i.e., $[x]^- = x$ if $x \le 0$ or $[x]^- =0$ otherwise. We further let
$$C_i := \frac{1}{\epsilon}-\Gamma_i,~\forall i,$$
so that after converting $Q_i(k)$ to be nonnegative by using $\tQ_i(k) = Q_i(k) + C_i$ for all $i$, everything is transformed back to the original online algorithm except that each $Q_i(k)$ is replaced by $\tQ_i(k)$, hence we can still show that the revenue achieved by this modified version of online algorithm is within $O(\epsilon)$ of the optimal revenue.

It might seem counter-intuitive that by letting $\epsilon$ go to zero, we can incur potentially large underdrafts (under-utilization of the budget) and yet are able to achieve maximum revenue. This is not a contradiction: for each fixed $\epsilon$,  in the long term, the average service provided to each client is close to the average budget. The $O(1/\epsilon)$ is a fixed amount by which the total budget up to any time $T$ is under-utilized, and, after divided by $T$, it goes to zero when $T$ approaches infinity.

\begin{figure}[t]
\centering
\includegraphics[width=9.0cm]{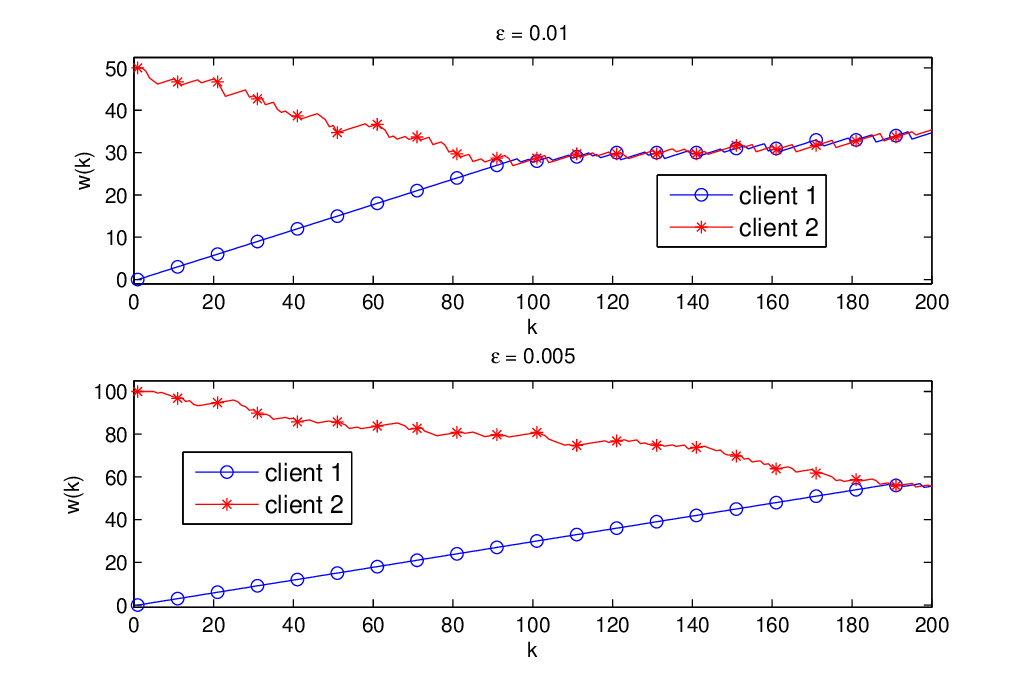}
\caption{Temporary unfairness in service} \label{fig: unfair}
\end{figure}

We note that while an underdraft does not seem to significantly hurt either the client, who actually benefits from an underdraft, or the service provider, whose long-run average revenue is still diminished only by $O(\epsilon)$, large values of the underdraft may result in temporary unfairness in the system.\footnote{Note that this temporary unfairness is not an artifact of the underdraft mechanism. In fact, it occurs once a sample path enters a state where some clients have huge differences from others in their corresponding queue lengths, which can also happen under the original algorithm. We are just using the underdraft scheme to illustrate this phenomenon.}
If, for example, a client accumulates a large underdraft compared to the other clients, then it may receive priority over other clients for large periods of time. To illustrate this, we consider an example with two clients and one queried keyword.  Assume that $\Gamma_i < 0$ for $i=1,2$, and at time slot $k_0$, $Q_1(k_0) = \Gamma_1$ and $Q_2(k_0) = -C_2$ (this occurs with a positive probability due to the ergodicity of the Markov chain $\{\QQ(k)\}$ proved before). We simulate the sample paths of the weights in the maximization step (\ref{eq: dyn_alg_choose_mat2}) with the following setting: budgets $b_1 = b_2 = 0.6$, click-through rates $c_1 = c_2 = 0.5$, revenue-per-click $r_1= r_2 =1$; the number of query arrivals per time slot equals $2$ w.p. $0.5$ and $0$ otherwise; a budgeting cycle equals to one time slot ($N=1$) for simplicity. The results for both $\epsilon = 0.01$ and $\epsilon = 0.005$ ($k = 0$ corresponds to $k_0$ here) are shown in Figure~\ref{fig: unfair}. Client $2$ keeps getting services until the weights of both clients reaches the same level, and the smaller $\epsilon$ is, the longer the ``unfair serving'' period lasts.

It should be mentioned that this underdraft idea can be used under any upper-bounded query arrival model, not restricted in the Bernoulli arrival model considered in this paper.


\section{A Universal Lower Bound on the Expected Overdraft Level} 
\label{sec: que_len_lb}
We want to show that in the long run, an expected overdraft level of $\Omega(\log(1/\epsilon))$ is unavoidable under any stationary ad assignment algorithm which achieves a long-term average revenue within $O(\epsilon)$ of the offline optimum, when the queue length is only allowed to be nonnegative. An ad assignment algorithm $\varpi$ is defined as a strategy which uses matrix $M^{\varpi}(t,q)\in \cM_q$ for ad assignment when a query for keyword $q$ arrives at each time slot $t$. During each budgeting cycle $k$, the revenue obtained from client $i$ under algorithm $\varpi$ is defined as
\begin{equation}
A^{\varpi}_i(k) \triangleq
    \sum_{t = kN}^{kN+N-1} \sum_{s} [M^{\varpi}(t, \tq(t))]_{is}\cdot \tc_{\tq(t)is}(t) \cdot r_{\tq(t)i}.  \label{eq: per-client-rev}
\end{equation}
We then define average revenue obtained from client $i$ per budgeting cycle as $\lambda^{\varpi}_i \triangleq E[A^{\varpi}_i(k)]$ in the steady state. The long-term average revenue (per time slot) is thus $\avgR^{\varpi} = \sum_{i} \lambda_i^{\varpi} / N$, and the overdraft level of client $i$ evolves as
\begin{eqnarray}
&& Q_i^{\varpi}(k+1) = \nonnegpart{Q_i^{\varpi}(k)+ A^{\varpi}_i(k) - \tb_i(k)}. \label{eq: overdraft_que}
\end{eqnarray}
Note that our online algorithm is one particular $\varpi$, which makes the decision based on the current overdraft levels of all clients.

To seek a universal lower bound on expected overdraft level in the long run (here, equivalent to steady state), we only have to consider those algorithms $\varpi$ such that $\avgQ^{\varpi}_i \triangleq E[Q_i^{\varpi}(k)] < \infty$ for all $i$. To categorize these ``stable'' algorithms, 
we define ``per-client revenue region,'' similar to the concept of ``capacity region'' in the context of queueing networks: \\

\begin{definition}[``Per-Client Revenue Region''] \label{def: cap_region}
$$
\cC \triangleq \!\left\{\blambda^{\varpi}\! = \!\{\lambda^{\varpi}_i\}\! \ge\! \mathbf{0}\!:\exists\varpi ~s.t.~\lambda_i^{\varpi} \triangleq E\left[A_i^{\varpi}(k)\right] \le b_i,~\forall i \right\},
$$
given fixed parameters $\{r_{qi}\}$, $\{b_i\}$, $\{c_{qis}\}$, $N$ and statistical properties of $\tq(t)$ and $\{\tc_{qis}(t)\}$.
\hfill $\diamond$
\end{definition}
The offline optimal average revenue is then equal to 
$\max_{\blambda \in \cC} \sum_{i} \lambda_i / N, $ which is denoted as $\avgR^*$.\\ 

Note that if the query arrival rates per budgeting cycle are too low, the average revenue drawn from some client will never hit its specified budget, no matter which algorithm $\varpi$ s.t. $\blambda^{\varpi} \in \cC$ you pick (i.e., $\exists~ i$ s.t. no feasible solution $\pp$ can make constraint \eqref{eq: budget} for this $i$ tight). The system resources (here, budgets) are underutilized and it is not so important to consider the tradeoff between revenue and overdraft. To avoid this, we can assume a relatively large $N$ (i.e., the number of time slots in one budgeting cycle) such that 
\begin{equation}\label{eq: large_N}
N \ge \max_i  \left\{ \frac{b_i}{\sum_q \nu_q r_{qi}\cdot \max_{M\in \cM^i_q}c_{qis(i,M)}} \right\},
\end{equation} 
where $\cM^i_q \subseteq \cM_q$ is defined as a set of ad assignment matrices, of which the $i^{\rm th}$ row has a ``1'', and $s(i,M)$ in $c_{qis(i,M)}$ refers to the column in $M$ where that ``1'' stays. This guarantees that for each $i$, there exists an algorithm $\varpi_i$ such that $\blambda^{\varpi_i} \in \cC$ and $\lambda^{\varpi_i}_i = b_i$. The reason is that

In the following text, we will assume the above condition for $N$.

\subsection{One Keyword, One Client and One Webpage Slot}
\label{sec: single_que_LB}
We start with the simplest model: one keyword, one client and one webpage slot (hence we omit all the subscripts in the corresponding notations). Under condition \eqref{eq: large_N}, the offline maximum average revenue is trivially $b/N$.\\

\begin{theorem}
\label{thm: single_que_LB}
Given a small $\epsilon >0$, if an algorithm $\varpi$ leads to $E[A^{\varpi}(k)] \ge b - \epsilon$ in the steady state, then 
$$\avgQ^{\varpi} \ge \frac{\log (1/\epsilon)}{2(1-\log(\varphi P_+))}-1,$$
where we assume that 
$$\varphi \triangleq \Pr(\mbox{no query arrival in a budgeting cycle}) >0,$$ 
and $P_+ \triangleq \Pr(\tb(k) > 0) >0$.\hfill$\diamond$
\end{theorem}
\

Note that this result works for any query arrival and budget spending model satisfying the above two stated assumptions, and not only restricted to the model we described in Subsection \ref{sec: stochastic}. 
In the proof below, we generally write $\tb(k)$ as a random variable which can possibly take all nonnegative integer values. 

\begin{proof}
We ignore the superscript $\varpi$ for brevity. 
The dynamics of the queue is rewritten as
$Q(k+1) = Q(k) + A(k) - \hb(k)$, 
where the actual departure process is defined as
\begin{equation}
\hb(k)\triangleq \left\{
\begin{array}{cl}
\tb(k) & \mbox{if}~~Q(k)+A(k)-\tb(k)\ge 0;\\
Q(k)+A(k) & \mbox{otherwise.} \label{eq: def_eff_serv}
\end{array}
\right.
\end{equation}
Let $p_i \triangleq \Pr(\hb(k)=i)$ and $q_i \triangleq \Pr(\tb(k)=i)$ in the steady state. Note that
\begin{eqnarray*}
b-\epsilon &\le& \!\!E[A(k)] = E[\hb(k)] =  \sum_{i=1}^{\infty}\Pr(\hb(k)\ge i) = \Pr(\hb(k)\ge 1) + \sum_{i=2}^{\infty}\Pr(\hb(k)\ge i)\\  
&\stackrel{(a)}{\le}&\!\!(1-p_0)+\sum_{i=2}^{\infty} \Pr(\tb(k)\ge i) 
=1-p_0 + \left(b - \Pr(\tb(k) \ge 1)\right) = 1-p_0+b-(1-q_0)\\
&=&\!\! q_0-p_0+b,
\end{eqnarray*}
where (a) holds because $\Pr(\hb(k) \ge i) \le \Pr(\tb(k)\ge i)$ for all $i\ge 0$. Thus, $p_0 \le q_0 + \epsilon$. Since
$$\Pr(\hb(k) = 0) = \Pr(\tb(k)=0) + \Pr(\hb(k)=0,~\tb(k)\ge 1),$$ 
we have $p_0 = q_0 + \tilde{p}_0$, where 
$\tilde{p}_0\triangleq \Pr(\hb(k)=0,~\tb(k)\ge 1)$. 
Therefore, 
\begin{equation}
\tilde{p}_0 \le \epsilon. \label{eq: zero_depart_UB}
\end{equation}

Next, we are looking for a lower bound on $\tilde{p}_0$ in relation to $\avgQ$. 
Letting $P_+ \triangleq \Pr(\tb(k) > 0)$ (which is surely a positive constant since $b >0$), we then have
\begin{eqnarray}
n \tilde{p}_0 &=& \sum_{k=0}^{n-1}\Pr(\hb(k) =0,~\tb(k)> 0)
\stackrel{(a)}{\ge} \Pr\left(\bigcup_{k=0}^{n-1} \{\hb(k)=0,~\tb(k)> 0\}\right)\nonumber\\
& \stackrel{(b)}{\ge}& \Pr(Q(0) \le n-1;~A(k) = 0, 
\tb(k)> 0,~\forall~ 0\le k \le n-1)\nonumber\\
&=& \Pr(Q(0) \le n-1) 
\cdot \prod_{k = 0}^{n-1}  \Pr(A(k)=0)\cdot\Pr(\tb(k)> 0) \nonumber\\
&=& (\varphi P_+)^{n}\cdot \Pr(Q(0) \le n-1) 
\stackrel{(c)}{\ge} (\varphi P_+)^n \left(1-\avgQ /n\right), \label{eq: np_0_LB}
\end{eqnarray}
where (a) holds according to the union bound, (b) holds since the event on the RHS implies the one on the LHS, 
and (c) holds due to the Markov inequality. If we pick 
$n := \left\lceil 2\avgQ\right\rceil \in [2\avgQ,  2\left(\avgQ+1\right)],$ 
inequality \eqref{eq: np_0_LB} further implies that
\begin{equation}
\tilde{p}_0 \ge \frac{(\varphi P_+)^n}{2n} \stackrel{(e)}{\ge} e^{-n(1-\log( \varphi P_+))} \ge e^{-2(\avgQ+1)(1-\log(\varphi P_+))}, \label{eq: zero_arr_LB}
\end{equation}
where (e) holds because $\frac{1}{2x} \ge e^{-x}$ for all $x> 0$. Combining inequalities \eqref{eq: zero_depart_UB} and \eqref{eq: zero_arr_LB} then completes the proof.
\end{proof}


\

In the related literature, \cite{BerGal_ToIT02_EnergyDelayTradeoff} comes up with an $\Omega(1/\sqrt{\epsilon})$ bound for a set of algorithms under some admissibility conditions,  
while \cite{Neely_TAC09_EnergyDelayTradeoff} provides an $\Omega(\log(1/\epsilon))$ bound for more general algorithms. 

Our proof uses the following ideas inspired by \cite{Neely_TAC09_EnergyDelayTradeoff}: if the throughput is lower bounded by a number close to the average potential departure rate, then the probability of zero actual departures given nonzero potential departures must be upper bounded by a small number; further, if the average queue length is given, then the probability of hitting zero must be upper bounded because otherwise, the queue length would become small. However, we cannot directly use the expression for the lower bound in \cite{Neely_TAC09_EnergyDelayTradeoff} since it  imposes certain strict convexity assumptions which do not apply to our model where the objective is linear. So we have provided a very simple derivation of the lower bound on the queue length for our specific model.

Additionally, our $\Omega(\log(1/\epsilon))$ bound based on a linear objective function can be extended to the multi-queue case (in Subsection~\ref{sec: multi_que_LB}). The $\Omega(1/\sqrt{\epsilon})$ bound in \cite{BerGal_ToIT02_EnergyDelayTradeoff} has been extended to the multi-queue case in \cite{neely2007optimal} 
but still under strict convexity assumption and for a restrictive class of algorithms. 
Whether the $\Omega(\log(1/\epsilon))$ bound in \cite{Neely_TAC09_EnergyDelayTradeoff} can be easily extended to multiple queues still remains a question.


\subsection{Multiple Keywords, Multiple Clients and Multiple Webpage Slots}
\label{sec: multi_que_LB}
We now extend this lower bound to the original general model, which can have multiple keywords, multiple clients and multiple webpage slots. 
It is easy to see that the ``per-client revenue region'' $\cC$ in Definition \ref{def: cap_region} is a polytope, 
which can then be rewritten as
\begin{equation}
\cC = \left\{\blambda \ge \mathbf{0}:~\sum_{i} h_i^{(n)}\lambda_i \le d^{(n)},~\forall~1\le n \le L \right\}, \label{eq: cap_region_hplane}
\end{equation}
where $h_i^{(n)} \ge 0$ and $d^{(n)} > 0$ for all $i$ and $n$.
The outer boundary of the polytope $\cC$ consists of the $L$ hyperplanes, i.e., $\sum_{i} h_i^{(n)}\lambda_i = d^{(n)}$ for all $n\in [1,L]$. 

Under condition \eqref{eq: large_N}, $L$ is at least equal to the number of clients (i.e., number of budget constraints), so \eqref{eq: cap_region_hplane} gives a more precise description of the stability condition for this ``multi-queue system,'' compared to the original definition of $\cC$. Thus, corresponding to the normal vector of each hyperplane, we convert the original multi-queue system into a new one with $L$ queues: 
For each $n \in [1,L]$, we first scale the $i^{\rm th}$ queue described in \eqref{eq: overdraft_que} by $h_i^{(n)}$, so that
it has a queue length equal to $h_i^{(n)}Q_i(k)$, with $h_i^{(n)} A_i(k)$ arrivals and $h_i^{(n)} \tb_i(k)$ potential departures in time slot $k$, for all $i$. Next, we treat $\sum_{i} h_i^{(n)} Q_i(k)$ as the $n^{\rm{th}}$  queue, and since any $\blambda \in \cC$ satisfies $\sum_{i} h_i^{(n)}\lambda_i \le d^{(n)}$, its maximum achievable average departure rate equals $d^{(n)}$, where
$d^{(n)} \le \sum_{i} h_i^{(n)}b_i$, 
because the potential departure rate of each individual scaled queue may not be fully achieved when all of them are coupled together. 

We then come up with the formal definition of the class of algorithms which achieves a ``near-optimal'' average revenue.\\

\begin{definition}[``\boldmath$\epsilon$\unboldmath-Neighbourhood'' of the maximum] Let $\blambda^*$ be one optimal point in $\cC$ such that  $\sum_i \lambda^*_i = \avgR^*$. 
The $\epsilon$-neighbourhood of $\blambda^*$ is defined as
\begin{equation}
\cN_{\epsilon} \triangleq \{\blambda^{\varpi}\in\cC \setminus \partial \cC:~0 < N\cdot (\avgR^*-\avgR^{\varpi}) \le \epsilon\}, \label{def: epsilon-neighbor}
\end{equation}
where $\partial \cC$ represents the outer boundary of $\cC$, and it should be noted that the average revenue is evaluated per time slot while $\blambda$ is evaluated per $N$ time slots.
\hfill $\diamond$
\end{definition}
\

Note that in the above definition, since $\blambda^{\varpi} \in \cN_{\epsilon}$ is not on any boundary, $\avgR^*$ is strictly larger than $\avgR^{\varpi}$, which is easy to see from some basic principles of linear programming.

The following theorem shows the universal lower bound  $\Omega(\log(1/\epsilon))$ for the general case.\\ 

\begin{theorem}
\label{thm: lb}
For any algorithm $\varpi$ s.t. $\blambda^{\varpi} \in \cN_{\epsilon}$, we have $$\sum_{i=1}^M \bar{Q}_i^{\varpi}  \ge  \frac{\log (1/\epsilon) - C_2}{C_1} -1,$$
where $\varphi \triangleq \Pr(\mbox{no query arrival in a budgeting cycle}) = (1-\nu)^N >0$, $P_+ \triangleq \Pr(\tb_i(k) > 0,~\forall i) > 0,$ and 
\begin{eqnarray}
C_1 &\triangleq& 2(1-\log(\varphi P_+))\cdot \max_{i,n}h^{(n)}_i \in (0,\infty),\nonumber\\
C_2 &\triangleq& \max\{\log(\max_{i,n} h^{(n)}_i),0\} \in [0,\infty). \label{eq: lb_const}
\end{eqnarray}
\hfill $\diamond$
\end{theorem}

\begin{proof}
We ignore the superscript $\varpi$ for brevity. According to some basic principles of linear programming, an optimal point $\blambda^*$ is at a corner of $\cC$. If there are several optimal points, any convex combination of them is also optimal. Denote this optimal point sets as $\Lambda^*$ and $\forall \lambda^*\in \Lambda^*$, $\exists~n^*\in [1,L]$, s.t. $\sum_{i} h_i^{(n^*)}\lambda_i^* = d^{(n^*)}$.

Given a $\blambda \in \cN_{\epsilon}$, $\exists~\btheta$ s.t. $\sum_{i} \theta_i = \sum_{i} \lambda_i^*$ and $\theta_i \ge \lambda_i$ for all $i$ (but at least one inequality is strict). Besides, for this $\btheta$, $\exists~\tn\in [1,L]$, s.t. $ \sum_{i} h_i^{(\tn)} \theta_i \ge d^{(\tn)}$ (otherwise, $\btheta \in \cC \setminus \partial \cC$ will hold and hence $\sum_{i} \theta_i < \sum_{i} \lambda_i^*$, which leads to a contradiction). Therefore, 
\begin{eqnarray}
d^{(\tn)} \!-\! \sum_{i} h_i^{(\tn)} \lambda_i \!\!\!&\le &\! \sum_{i} h_i^{(\tn)} (\theta_i - \lambda_i) \stackrel{(a)}{\le}  h^{(\tn)}_{max} \sum_{i} (\theta_i - \lambda_i) \nonumber \\
&=&  h^{(\tn)}_{max} \sum_{i} (\lambda^*_i - \lambda_i) \le  h^{(\tn)}_{max}\epsilon, \label{eq: service_surplus}
\end{eqnarray}
where $h^{(\tn)}_{max} \triangleq \max_{i} h^{(\tn)}_i>0$ and inequality (a) holds because $\theta_i \ge \lambda_i$ for all $i$.
Letting $P'_+ \triangleq \Pr(\sum_{i} h_i^{(\tn)} \tb_i(k) > 0)$, it is easy to see that $P'_+ \ge \Pr(\tb_i(k) > 0,~\forall i) = P_+ >0$. Together with Theorem \ref{thm: single_que_LB}, we can conclude that 
\begin{eqnarray}
\sum_{i} h_i^{(\tn)}\bar{Q}_i 
&\ge & \frac{\log (1/\epsilon)-\log(h^{(\tn)}_{max})}{2(1-\log(\varphi P'_+))}-1
\ge \frac{\log (1/\epsilon)-\log(h^{(\tn)}_{max})}{2(1-\log(\varphi P_+))}-1.
\nonumber
\end{eqnarray}
Since $\sum_{i} h_i^{(\tn)}\bar{Q}_i \le h^{(\tn)}_{max} \sum_{i} \avgQ_i$, it is further concluded that
\begin{equation*}
\sum_{i} \avgQ_i \ge 
 \frac{\log (1/\epsilon)-\log(h^{(\tn)}_{max})}{2 h^{(\tn)}_{max}(1-\log(\varphi P_+))}-1 
\ge \frac{\log (1/\epsilon) - C_2}{C_1} - 1,
\end{equation*}
where the universal constants are defined in \eqref{eq: lb_const}, 
and it is guaranteed that $C_1 \in (0,\infty)$ and $C_2 \in [0,\infty)$. This completes the proof. 
\end{proof}

\begin{figure}[t]
\centering
\includegraphics[width=8.7cm]{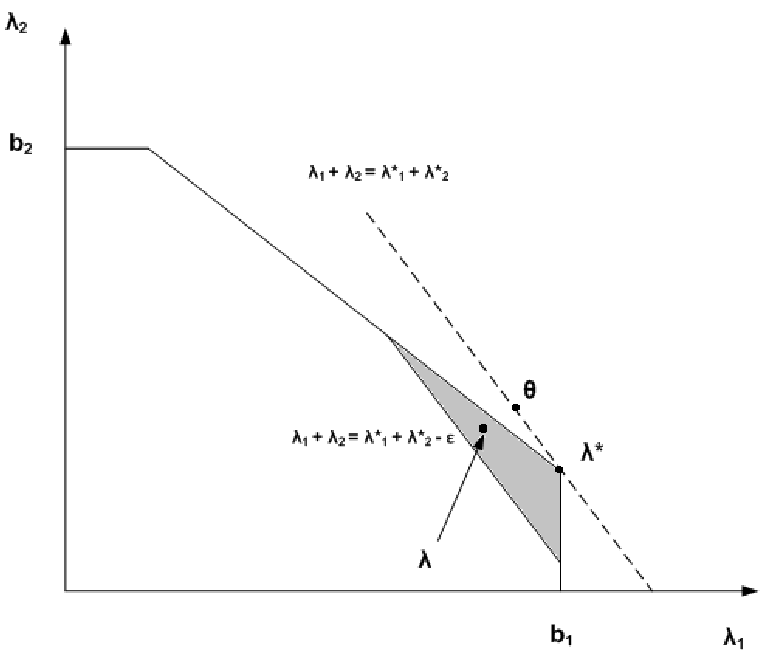}
\caption{An illustration of the idea in the proof of Theorem~\ref{thm: lb}} 
\label{fig: 2D_proof}
\end{figure}
\

\begin{remark}
We briefly explain the idea behind choosing $\theta$ in the above proof: For those $\blambda\in \cN_{\epsilon}$ such that $\lambda_i  \le \lambda_i^*$ for all $i$ (at least one is strict), $\btheta$ can be directly chosen as $\blambda^*$ to make inequality (a) in (\ref{eq: service_surplus}) hold. But for the other $\blambda\in \cN_{\epsilon}$ which do not satisfy the above condition, it is necessary to introduce a $\btheta$ other than $\blambda^*$, which both lies on the ``maximum revenue line'' (i.e., $\sum_{i} \theta_i = \sum_{i} \lambda_i^*$) and dominates $\blambda$ component-wise, in order to derive inequality~(\ref{eq: service_surplus}). Note that $\btheta$ is not unique and furthermore, $\btheta$ lies either on $\partial \cC$ or in the exterior of $\cC$ and it can be chosen as a boundary point only if the optimal revenue point is not unique. Figure~\ref{fig: 2D_proof} illustrates this idea using an example with one keyword, two clients and one webpage slot, 
specifically for showing where such a $\btheta$ is located. 
\hfill $\diamond$
\end{remark}

The basic idea in our proof is to use Theorem \ref{thm: single_que_LB} to first get a lower bound for those new single queues written as a ``weighted sum'' of the original queues (described above). This idea is similar to one part in the proof for the lower bound on the expected queue length of a departure-controlled multi-queue system in \cite{shah_Allerton08_LB}, but some technique in their proof cannot directly apply to arrival-controlled queues like ours.

\subsection{Tightness of the Lower Bound}
\label{sec: LB_tightness}
We want to show that the $\Omega(\log(1/\epsilon))$ universal lower bound is tight, i.e., achievable by some algorithms. Consider the following simple queueing model: the arrival process $a(k)$ is i.i.d. across time, $a(k) = 2$ w.p. $\nu$ and $a(k)=0$ otherwise. The service rate is constant and equal to $1$. Assume that $ \nu \in (1/2, 1)$. With the controlled arrival process $\ha(k)$, we want to achieve a throughput $E[\ha(k)] \ge 1-\epsilon$ for a given small $\epsilon>0$. A ``threshold policy'' based on a threshold $T$ is proposed below:
\begin{itemize}
\item When $Q(k) > T$, reject all arrivals;
\item When $Q(k) = T$, accept one arrival w.p. $p_1$, accept two arrivals w.p. $p_2$, and reject all of them otherwise. 
\item When $Q(k) < T$, accept all arrivals.
\end{itemize} 
Defining $\pi_i$ as the steady-state probability that $Q(k) =i$ ($0\le i \le T+1$) for the resulting Markov chain, the local balance equations are given below:
\begin{eqnarray}
\pi_i \nu &= &\pi_{i+1} (1-\nu),~\forall~0\le i\le T-2; \nonumber\\
\quad \pi_{_{T-1}} \cdot p_1\nu  &=& \pi_{_{T}}(1-(p_1+p_2)\nu);\nonumber\\
\quad \pi_{_{T}} \cdot p_2\nu &=& \pi_{_{T+1}};\nonumber\\
\quad \sum_{i=0}^{T+1} \pi_i &=& 1. \label{eq: local_balance}
\end{eqnarray}
Combining these equations with the throughput requirement, we get  
\begin{equation}
\nu \left[ 2\sum_{i=0}^{T-1} \pi_i + \pi_{_T} (2p_2+p_1)\right] = 1 - \epsilon, \label{eq: required_throughput_eq}
\end{equation}
and one can finally show that (ignoring detailed calculations)
$$
T = 
\frac{\log(1/\epsilon)+\log C(\epsilon)}{\log \left(\frac{\nu}{1-\nu}\right)},$$
where $$C(\epsilon) \triangleq \frac{(2\nu-1+\epsilon)(1-\nu(p_1+p_2))}{\nu(2-2(1-\nu)p_2 - p_1)}.
$$
The above result further implies that $\avgQ \sim \Theta(\log(1/\epsilon))$. we can also see that as $\nu \rightarrow 1$, $T\rightarrow 0$, which is consistent with the fact the lower bound given in Theorem \ref{thm: single_que_LB} goes to $0$ as the ``zero arrival probability'' $\varphi \rightarrow 0$.  


Another example showing the tightness of an $\Omega(\log(1/\epsilon))$ bound is the dynamic packet dropping algorithm in \cite{Neely_TAC09_EnergyDelayTradeoff} (note that this universal lower bound is proved based on a strict convexity assumption as mentioned before in Subsection \ref{sec: single_que_LB}).

\section{Click-Through Rate Maximization Problem}
\label{sec: click-through}
In this section, we consider another online ads model, in which the objective is to maximize the long-term average total click-through rate of all queries. Instead of average budget, client $i$ specifies in the contract an average ``{\em impression requirement}'' $m_i$, which is the minimum number of times an ad of this client should be posted by the service provider per ``{\em requirement cycle}'' (equal to $N$ time slots) on average. The other parameters are the same as in the model proposed in Section \ref{sec: intro} for the revenue maximization problem. 

The corresponding optimization formulation now becomes
\begin{equation}
\max_{\pp \in \mathcal{F}} \avgJ(\pp) = \sum_q\nu_q  \sum_{M\in \cM_q}\pqM \sum_{i,s} M_{is} c_{qis} \label{obj: max_click-rate}
\end{equation}
where the feasible set $\mathcal{F}$ is characterized by  
\begin{eqnarray}
N \sum_q\nu_q  \sum_{M\in \cM_q}\pqM \sum_{s} M_{is} \ge m_i,\qquad \forall i;\label{eq: obligation}\\
0\leq \pqM \leq 1,\qquad \forall q,~M \in \cM_q;\label{eq: prob2}\\
\sum_{M\in \cM_q} \pqM \leq 1, \qquad \forall q. \label{eq: prob_sum2}
\end{eqnarray}

Different from the revenue maximization problem, here the feasible set can become empty if some $m_i$ is too high. Basically, without constraint \eqref{eq: obligation}, $\mathcal{F}$ is relaxed to 
\begin{equation}
\mathcal{F}_0 \triangleq \{\pp: 0\leq \pqM \leq 1,~\forall q,M \in \cM_q; \displaystyle\sum_{M\in \cM_q} \pqM \leq 1, ~\forall q\}.\label{eq: feasible_relaxed}
\end{equation}
We can then define the following capacity region which characterizes how large the average number of impressions can be achieved for each client per requirement cycle:
$$
\cC \triangleq \left\{
\bmu: \mu_i \! =\! N \!\sum_q\nu_q \!\!\!\! \displaystyle \sum_{M\in \cM_q}\pqM \!\!\sum_{s} M_{is},~\forall i,~s.t.~\pp \in \mathcal{F}_0
\right\}.
$$
Clearly, $\mm\in \cC$ must hold to ensure the existence of a solution for the above optimization problem. 

Through a similar approach as in Subsection \ref{sec: iterative_sol}, 
we can write down a similar online algorithm based on the same stochastic model as defined in Subsection \ref{sec: stochastic}. 
We define $\qq(k) \triangleq \{\tq(t),~\mbox{for}~kN \le t \le kN+N-1\}$. Similar to $\tb_i(k)$, $\tm(k) = \lceil m_i \rceil$ w.p. $m_i - \lfloor m_i \rfloor$ and $\tm(k) = \lfloor m_i \rfloor$ otherwise.
   
\noindent\hrulefill\\
\textbf{Online Algorithm:} (in each requirement cycle $k\ge 0$)

In each time slot $t \in [kN, kN+N-1]$, if $\tq(t) > 0$, choose the assignment matrix
\begin{eqnarray}
&&\tM^*(t,\tq(t),\QQ(k)) 
\in \arg \max_{M\in \cM_{\tq(t)}} \sum_{i,s} M_{is} \left(\frac{c_{\tq(t) is}}{\epsilon}+Q_i(k)\right). \label{eq: dyn_alg_choose_mat2}
\end{eqnarray}

At the end of requirement cycle $k$, for each client $i$, update
$$Q_i(k+1) = \nonnegpart{Q_i(k) + \tm(k) - S_i(k,\QQ(k),\qq(k))}, 
$$
where
\begin{equation}
S_i(k,\QQ(k),\qq(k)) \triangleq
    \sum_{t = kN}^{kN+N-1} \sum_{s} [\tM^*(t, \tq(t), \QQ(k))]_{is}.  \label{eq: virtual_depart}
\end{equation}

\noindent\hrulefill\\

In real online advertising business, some clients may only have short-term contracts, 
i.e., clients may not be interested in the average number of impressions per time slot but may be interested in a minimum number of impressions in a given duration (such as a day). Further, query arrivals may not form a stationary process. In fact, they are more likely to vary depending on the time of day. These extensions are considered in Appendix E. Such extensions also make sense for the revenue maximization model considered in the previous sections, but the approach is similar to Appendix E and so will not be considered here.

\subsection{Performance Evaluation} 
\label{sec: perf_model2}

$S_i(k,\QQ(k),\qq(k))$ defined in \eqref{eq: virtual_depart} represents the actual number of impressions for client $i$'s ads during requirement cycle $k$. The queue length increases when the average impression requirements in a particular requirement cycle cannot be fulfilled. Hence, a positive queue represents accumulated ``credits,'' which enhances the chance of being assigned with a webpage slot in the future, much like a negative queue in the revenue maximization problem. 
We thus call this queue a ``{\em credit queue}.'' 



Unlike the revenue maximization problem in which an $O(1/\epsilon)$ upper bound on the transient queue length is automatically imposed by the online algorithm, here we need to prove the stability of the queues and show an upper bound on the mean queue length. Since $\{\QQ(k)\}$ defines an irreducible and aperiodic Markov chain, in order to prove its stability (positive recurrence), we will first bound the expected drift of $\QQ(k)$ for a suitable Lyapunov function.\\

\begin{lemma}
\label{lemma: bounded_drift_model2}
Consider the Lyapunov function $V(\QQ) = \frac{1}{2} \sum_{i} Q_i^2$. For any $\epsilon>0$ and each requirement cycle $k$,
\begin{eqnarray}
&&E[V(\QQ(k+1))|\QQ(k)=\QQ] - V(\QQ)
\le \frac{D_3}{\epsilon}  + D_1
    - D_2 \sum_i Q_i. \label{eq: neg_drift}
\end{eqnarray}
Here,
\begin{eqnarray}
D_1 \!\!\!&\triangleq&\!\!\! \frac{1}{2} \Big(N(N-1) L^2+NL
 + \sum_i \lceil m_i \rceil^2 (m_i- \lfloor m_i \rfloor) 
+ \lfloor m_i \rfloor^2 (1-m_i+\lfloor m_i \rfloor)
\Big), \label{eq: D1_def}
\end{eqnarray}
where $L$ is the number of webpage slots; 
\begin{equation}
D_2\triangleq \min_{i} \{N\sum_q \nu_q \sum_{M\in\cM_q} \hpqM \sum_{s} M_{is}-m_i\}, \label{eq: D2_def}
\end{equation}
for some $\hpp\in\mathcal{F}$ such that $D_2>0$
; and 
\begin{equation}
D_3 \triangleq N \cdot \max_{\pp\in \mathcal{F}_0}\avgJ(\pp) \label{eq: D3_def}
\end{equation} 
where $\mathcal{F}_0$ is defined in \eqref{eq: feasible_relaxed}. 
\hfill $\diamond$
\end{lemma}
\

The proof is similar to the proof of Lemma \ref{lemma: bounded_drift2} with some modifications in the final steps, which will be given briefly in Appendix~\ref{sec: proof_bounded_drift_model2}. With this lemma, we can conclude that $\QQ(k)$ is positive recurrent because the expected Lyapunov drift is negative except for a finite set of values of $\QQ(k)$, according to Foster-Lyapunov theorem (\cite{asmussen03AP_que, MeynTweedie93markov}). \\

\begin{remark}
Note that compared to the definition of $B_2$ in \eqref{eq: B2_def} of Lemma \ref{lemma: bounded_drift2} where $B_2\ge 0$, $D_2$ needs to be strictly positive in order to prove the stability of queues. Such a $\hpp$ in the definition of $D_2$ can always be found unless $\mathcal{F}$ is a degenerate set with at most one element. \hfill$\diamond$
\end{remark}

The stability of the queues directly implies the following corollary: 
\begin{corollary}[Overservices in the long term]
\label{cor: long-run obligation}
$$\lim_{K\rightarrow\infty} E\left[\frac{1}{K}\sum_{k=1}^{K} S_i(k,\QQ(k),\qq(k))\right] \ge m_i,~~\forall i.$$
\hfill $\diamond$
\end{corollary}

In addition to proving stability, Lemma \ref{lemma: bounded_drift_model2} will be used to evaluate the upper bound on the expected total queue length in the steady state, as shown in the following theorem: \\

\begin{theorem}
\label{thm: qlen_ub_model2}
Under the online algorithm,
\begin{equation}
E\left[\sum_{i} Q_i(\infty)\right] \le \frac{1}{D^*_2}\left(D_1 + \frac{D_3}{\epsilon}\right),
\label{eq: qlen_ub_model2}
\end{equation}
where $D_1$ and $D_3$ are respectively defined in \eqref{eq: D1_def} and \eqref{eq: D3_def}
; $D^*_2$ is defined as
\begin{equation}
D^*_2 \triangleq \max_{\pp\in \mathcal{F}_0} D_2(\pp). \label{eq: max_D2}
\end{equation}  
where $D_2$ is defined in \eqref{eq: D2_def} (regarded as a function of $\pp$).
\hfill $\diamond$
\end{theorem}
\

\begin{proof}
Averaging both sides of inequality \eqref{eq: neg_drift} over $0\le k\le K-1$, taking $K\rightarrow\infty$ and doing some simple algebra, one obtains
$$
\limsup_{K\rightarrow\infty} \frac{1}{K}\sum_{k=0}^{K-1} E\left[\sum_{i} Q_i(k)\right] \le \frac{1}{D_2}\left(D_1 + \frac{D_3}{\epsilon}\right).
$$
The LHS equals to $E\left[\sum_{i} Q_i(\infty)\right]$ according to Theorem 15.0.1 in \cite{MeynTweedie93markov}. The RHS is minimized through maximizing $D_2$ over all $\pp \in \mathcal{F}_0$ (which will certainly satisfy $\pp \in \mathcal{F}$ and $D_2 >0$). This completes our proof. 
\end{proof}
\

The following theorem shows that the online algorithm proposed above achieves a long-term average click-through rate within $O(\epsilon)$ of the offline optimum. The proof is similar to the one for Theorem \ref{thm: revenue_converge} and hence will be omitted. 
\begin{theorem}
\label{thm: ctr_converge}
For any $\epsilon >0$,
$$
0 \le \lim_{K\rightarrow \infty} E\left[ \bar{J}(\pp^*) - \frac{1}{K\!N}\sum_{k=0}^{K-1} J(k) \right] \le \frac{D_1\epsilon}{N},
$$
for some constant $D_1 > 0$ (defined in \eqref{eq: D1_def} in Lemma~\ref{lemma: bounded_drift_model2}). Here, $J(k)$ is defined as the total number of click-through events within requirement cycle $k$.\hfill $\diamond$
\end{theorem}

\subsection{Customizing Impression Requirements $\{m_i\}$ Based on Query Arrival Rates $\{\nu_q\}$} 
\label{sec: choose_impression} 
Since a positive queue measures how much the service provider ``owes'' a client, reducing the coefficient of the $1/\epsilon$ term in the upper bound on the mean queue length becomes important. Besides, we also need to guarantee $\mm\in \cC$. In order to handle these two issues, we introduce an approach to customizing $\{m_i\}$ based on known (or estimated) query arrival rates $\{\nu_q\}$,
 
Replacing $D^*_2$ in Theorem \ref{thm: qlen_ub_model2} by a common $D_2$ defined in equation~\eqref{eq: D2_def}, if we want the expected total queue length to be upper bounded by $Q_{max}$, it suffices to let  
\begin{equation}
D_2 \ge \xi \triangleq \frac{1}{Q_{max}}\left(D_1+\frac{D_3}{\epsilon}\right), \label{eq: bound_que_coef}
\end{equation}
where $D_3$ is already determined, and $D_1$ does not matter much given a small $\epsilon$ although it includes unknown $\{m_i\}$. We then solve the following optimization problem to determine $\{m_i\}$:
\begin{eqnarray*}
&&\max_{\pp\in \mathcal{F}_0, \mm} \sum_i \log m_i \\
&\mbox{s.t.}& N\sum_q \nu_q \sum_{M\in\cM_q} \pqM \sum_{s} M_{is}-m_i \ge \xi,~\forall i.
\end{eqnarray*}
Here we use $\sum_i \log m_i$ as the objective function in order to guarantee a unique optimal solution and impose a certain fairness rule called ``proportional fairness'' (see e.g. \cite{kelly1998rate}). Note that $\xi$ cannot be set too large (i.e., $Q_{max}$ cannot be set too small), otherwise there may not exist a feasible solution.

Naturally, a question would arise: now that we need to solve some mathematical programming like the above one based on knowledge of query arrival rates, why not also directly solve the original linear programming in \eqref{obj: max_click-rate} and use the offline optimal solution $\pp^*$ to assign ads? 
The answer to this is similar to the max-weight algorithm for wireless networks. In \cite{shakkottaiTAC08effective} and \cite{yingToIT06LargeDev}, it has been shown that adaptive algorithms lead to much better queueing performance compared to static offline algorithms. We verify this assertion in our context through simulations in the next subsection.

\subsection{Queue Update in a Faster Time Scale}
\label{sec: fast_que}
In the original algorithm, the queue length is updated only at the end of each requirement cycle and used in the max-weight matching for the next whole requirement cycle. The longer a requirement cycle lasts, the more obsolete the queue length information becomes, so with a large $N$, short-term performances may not be so good even if long-term performances are still guaranteed. 

We then propose a solution which updates queue lengths in a faster time scale. Specifically, we divide each requirement cycle into $T$ {\em queueing cycles} with equal lengths (assuming $N/T \in \mathcal{Z}^+$ without loss of generality). We use $\{\hQQ(k,\tau):~0\le \tau \le T\}_{k\ge 0}$ to denote this new queueing system and assume $\hQQ(-1,T) = \mathbf{0}$. At the beginning of each requirement cycle $k$ before any decision, update 
$$\hQQ(k,0) = \hQQ(k-1,T) + \mathbf{\tm}(k),$$ 
and at the end of the $\tau^{\rm th}$ queueing cycle within this requirement cycle ($1\le \tau \le T$), for all client $i$,
\begin{eqnarray*}
&&\hQ_i(k,\tau)
 =\nonnegpart{\hQ_i(k,\tau\!-\!1) - \!\!\!\!\sum_{t = kN+(\tau-1)\frac{N}{T}}^{kN+\tau\frac{N}{T}-1} \sum_{s} [\tM^*(t, \tq(t), \hQQ(k,\tau))]_{is}}\!\!\!\!.
\end{eqnarray*}
Since $\|\hQQ(k,T) - \QQ(k)\| \le B$ for some constant $B$ independent of the queue lengths, it can be shown that the long-term performances evaluated in Subsection \ref{sec: perf_model2} are still guaranteed (the idea behind such a proof would be similar to the one in \cite{eryilmazToN05stable} and so is omitted).

Next, we use simulations to compare three different algorithms, namely a randomized algorithm following the offline optimal solution (labeled as OPT) and two versions of our online algorithm ``max-weight matching'' with and without ``fast queue update'' respectively (labeled as MWM-Fast and MWM respectively). In each scenario we test, all the parameters are randomly generated. The impression requirements $\{m_i\}$ are  chosen through the approach in Subsection \ref{sec: choose_impression}. 

\begin{figure}[ht]
\centering
\subfigure[Over-Service]{
\includegraphics[width=4.1cm]{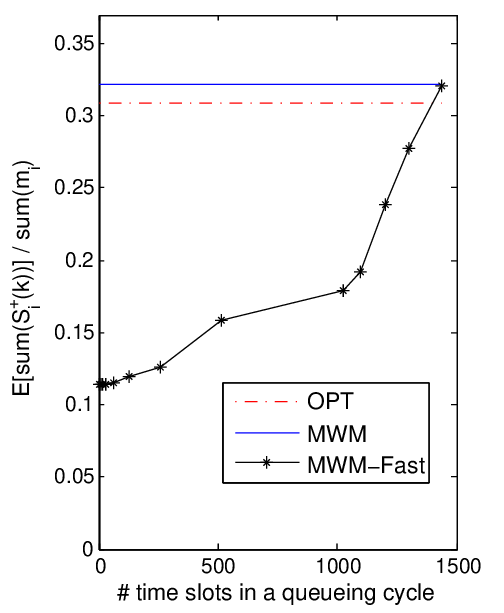}
\label{fig: overserv}
}
\subfigure[Under-Service]{
\includegraphics[width=4.1cm]{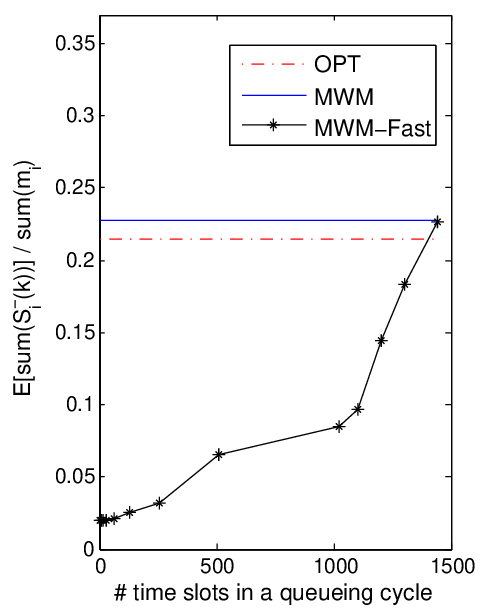}
\label{fig: underserv}
}
\caption{Average overall over-service and under-service (normalized by the total impression requirement) impacted by the ``fast queue update''}
\end{figure}

\begin{figure}[ht]
\centering
\subfigure[Over-Service]{
\includegraphics[width=4.1cm]{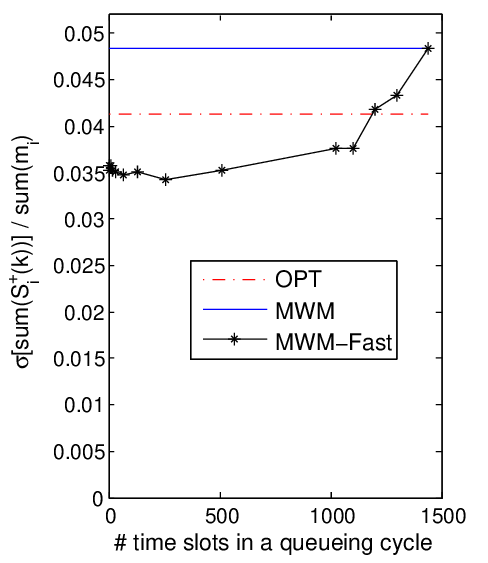}
\label{fig: overserv_std}
}
\subfigure[Under-Service]{
\includegraphics[width=4.1cm]{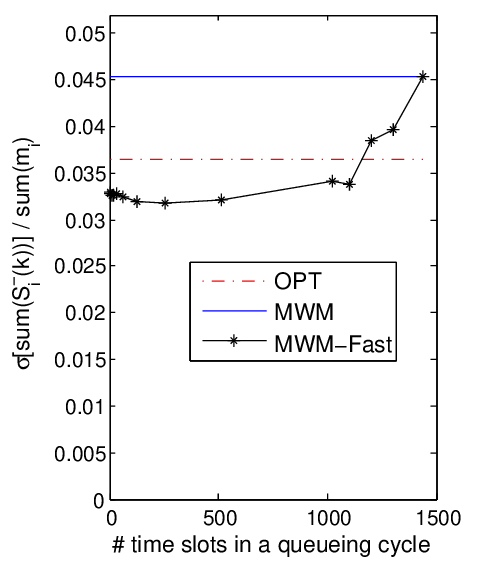}
\label{fig: underserv_std}
}
\caption{The standard variance of overall over-service and under-service (normalized by the total impression requirement) impacted by the ``fast queue update''}
\end{figure}

%

We take an example scenario with 2 webpage slots, 5 keywords and 10 clients. 
The probability that a query arrives in a time slot equals $0.7$. Specifically, for the five keywords, the query arrival rates are $\bnu = [0.2364, 0.0594, 0.1669, 0.0714, 0.1659]$. Table \ref{tab: ctr} shows the click-through rates for the ten clients ($C_1 \sim C_{10}$) corresponding to each keyword ($q_1 \sim q_5$), on webpage slots 1 and 2 respectively (a zero click-through rate indicates that the corresponding client is not related to this keyword). We use $N = 1440$ (say, one time slot is one minute and one requirement cycle is one day), $\epsilon = 10^{-4}$ and $Q_{max} = 20/\epsilon$ (recall that $Q_{max}$ is used to set up an upper bound on the mean queue length by the heuristic in  Subsection \ref{sec: choose_impression}). The simulation has been run for 1000 requirement cycles.

\begin{table}
\setlength{\tabcolsep}{2.5pt}
\begin{center}
  \begin{tabular}{|c||c|c | c | c |c| c | c |c| c | c |c| }
        \hline 
   &$C_1$&$C_2$&$C_3$&$C_4$&$C_5$&$C_6$&$C_7$&$C_8$&$C_9$&$C_{10}$\\
    \hline\hline
 &  \multicolumn{10}{|c|}{Webpage Slot 1}\\ \hline
    $q_1$ &  0 & 	0.519	&0.973&	0	&0.649&	0&	0&	0&	0.800&	0\\ \hline
$q_2$ & 0&	0&0&	0.340&	0	&0	&0.952&	0&	0&	0 \\ \hline
$q_3$ & 0.982&	0.645&	0.856&	0.461&	0.190&	0&	0.369&	0.669&	0.156&	0\\ \hline
$q_4$ & 0.423&	0&	0&	0&	0.599&	0&	0.179&	0&	0.471&	0.094 \\ \hline
$q_5$ & 0&	0&	0&	0.875&	0&	0.518&	0&0&	0&	0 \\
\hline \hline

   &  \multicolumn{10}{|c|}{Webpage Slot 2}\\
        \hline 
    $q_1$ &  0&	0.235&0.421&	0&	0.536&	0	&0	&0	&0.067	&0\\  \hline
    $q_2$ &  0&	0&	0&	0.312&0&	0&	0.050&0&	0&	0\\ \hline
    $q_3$ &  0.118	&0.248&	0.194& 0.222&	0.036&	0&	0.158&0.252&	0.092& 0 \\ \hline
    $q_4$ &  0.296	&0	&0	&0	&0.020	&0	&0.124&	0&	0.032&	0.060 \\ \hline
    $q_5$ &  0&	0&	0&	0.826&	0&	0.330&	0&	0&	0&	0\\
    \hline
\end{tabular}

\end{center}
 \caption{Click-through rates for all the clients' ads}
\label{tab: ctr}
\end{table}

To compare the performances of all the three algorithms, instead of considering the long-term performance requirements that we have used in the theory, we introduce two new metrics: over-service $S^+_i(k)\triangleq \nonnegpart{S_i(k) - \tm_i(k)}$ and under-service $S^-_i(k)\triangleq \nonnegpart{\tm_i(k)- S_i(k)}$ to client $i$ during requirement cycle $k$. 
Note that these metrics measure deviations from the guarantees over short time scales and so are more stringent requirements than the long-term guarantees used in the theory.

We show respectively in Figures \ref{fig: overserv} and \ref{fig: underserv} that the average overall over-service and under-service normalized by the total impression requirement, i.e., $E[\sum_i S^+_i (k)]/\sum_i m_i$ and $E[\sum_i S^-_i (k)]/\sum_i m_i$, 
are both reduced by the fast queue update.  Similarly, a ``variance reduction'' effect is shown by the fast queue update based on the statistics $\sqrt{var[\sum_i S^+_i(k)]} / \sum_i m_i$ and $\sqrt{var[\sum_i S^-_i(k)]} / \sum_i m_i$, respectively in Figures \ref{fig: overserv_std} and \ref{fig: underserv_std}. In  terms of the overall click-through rate, our simulation has verified that the three algorithms achieve approximately the same performance (the figure is omitted here) and further demonstrated in Figure \ref{fig: ctr_std} that the fast queue update can also reduce its variance. Note that these performances of each individual client also improve and we simply omit the figures here. 

Observed from Figure \ref{fig: que}, the offline optimal solution leads to very unstable queue dynamics. This essentially arises from the fact that the algorithm operates on an optimal point $\pp^*$ for which some inequalities in constraint \eqref{eq: obligation} may be tight. In contrast, our online algorithm guarantees the stability of queues, and the faster the queues update, the more stable the queue dynamics become (as an example we use $T=24$, i.e., the number of time slots per queueing cycle equals $60$). This is consistent with the above results which show a reduction of over-service and under-service in both mean and variance since these metrics directly measure the level of deviations around the equilibrium point of each stable queue. \\

\begin{remark}
While a long-term client may only be concerned with average performances, a short-term client cares about both mean (the average level for all the clients of its type) and variance (related to its own individual level), especially for the performances of under-service and click-through rate.\footnote{Over-service are cared about by the online ads service provider.} All of these are well handled by our online algorithm with fast queue updates. 
\end{remark}

\begin{figure}[t]
\centering
\includegraphics[width=8.7cm]{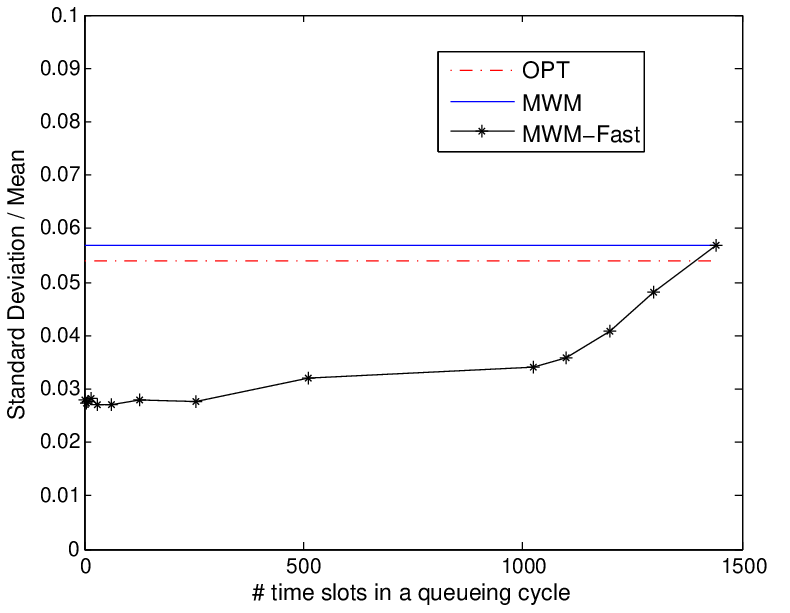}
\caption{The ``standard variance to mean ratio'' of overall click-through rate impacted by the ``fast queue update''} 
\label{fig: ctr_std}
\end{figure}

\begin{figure}[t]
\centering
\includegraphics[width=8.7cm]{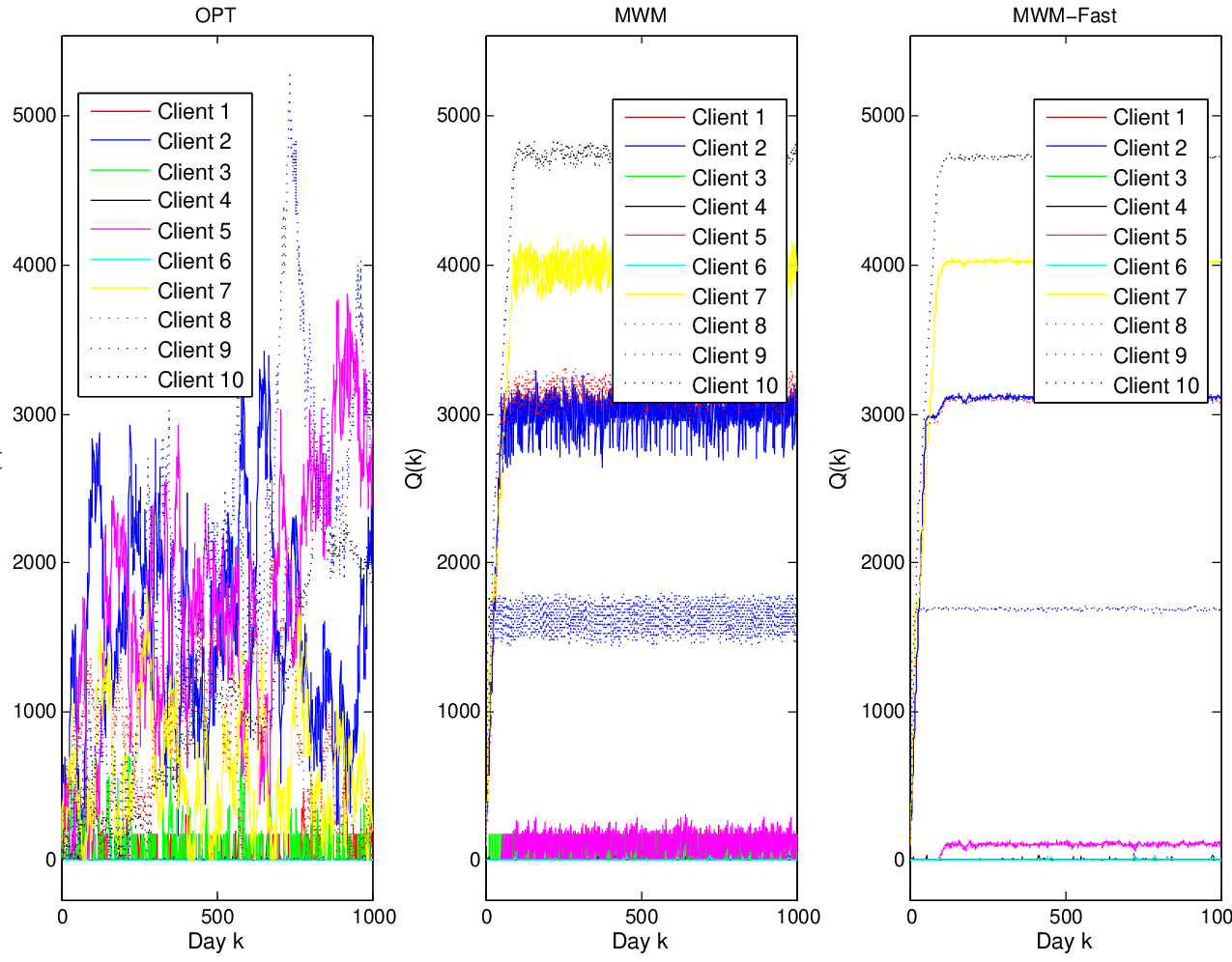}
\caption{Queue dynamics under three algorithms} 
\label{fig: que}
\end{figure}

\section{Conclusions}
\label{sec: conclusion}
In this paper, we propose a stochastic model to describe how search service providers charge client companies based on users' queries for the keywords related to these companies' ads by using certain advertisement assignment strategies. We formulate an optimization problem to maximize the long-term average revenue for the service provider under each client's long-term average budget constraint, and design an online algorithm which captures the stochastic properties of users' queries and click-through behaviors. We solve the optimization problem by making connections to scheduling problems in wireless networks, queueing theory and stochastic networks. Our online algorithm is entirely oblivious to query arrivals and fully adaptive, so even non-stationary query arrival patterns and short-term clients can be handled.

With a small customizable parameter $\epsilon$ which is the step size used in each iteration of the online algorithm, we have shown that our online algorithm achieves a long-term average revenue which is within $O(\epsilon)$ of the optimal revenue and the overdraft level of this algorithm is upper bounded by $O(1/\epsilon)$. By allowing negative values for the length of overdraft queues, we can eliminate overdraft.

When estimated click-through rates instead of true ones are used in our online algorithm, we show that the achievable fraction of the offline optimal revenue is lower bounded by $\frac{1-\Delta}{1+\Delta}$, where $\Delta$ is the relative error in click-through rate estimation.

We also show that in the long run, an expected overdraft level of $\Omega(\log(1/\epsilon))$ is unavoidable (a universal lower bound) under any stationary ad assignment algorithm which achieves a long-term average revenue within $O(\epsilon)$ of the offline optimum. The tightness of this universal lower bound is also shown for a simple queueing model using a threshold policy. 

In another optimization formulation where the objective is to maximize the long-term average click-through rate and the constraints include a minimum impression requirement for each client, we further propose an approach to set impression requirements which make the contract feasible and limit the average accumulated under-service to clients. 
Simulations show that making queues update in a faster time scale will reduce both over-service and under-service, which benefits a system involving short-term clients.


\bibliographystyle{abbrv}
\bibliography{adwords}

\appendix
\subsection{Proof of Lemma~\ref{lemma: bounded_drift2}}
\label{sec: proof_bounded_drift2}
\begin{eqnarray}
&&E[V(\QQ(k+1))|\QQ(k)=\QQ] - V(\QQ) \nonumber\\
&=& \frac{1}{2} E\left[\sum_i \left( \nonnegpart{Q_i+A_i(k,\QQ,\uu(k))-\tb_i(k)}\right)^2-Q_i^2\right]  \nonumber\\
&\le&  \frac{1}{2} E\left[\sum_i \left(Q_i+A_i(k,\QQ,\uu(k))-\tb_i(k)\right)^2 - Q_i^2\right]\nonumber\\
&=& E\Bigg[\sum_i Q_i \left(A_i(k,\QQ,\uu(k))-\tb_i(k)\right) 
+ \frac{1}{2}\sum_i \left(A_i(k,\QQ,\uu(k)) - \tb_i(k)\right)^2\Bigg]\nonumber \\
&\le& \sum_i Q_i \left(\lambda_i(k,\QQ)-b_i\right) 
+ \frac{1}{2} \sum_i ( E[A_i^2(k,\QQ,\uu(k))] + E[\tb_i^2(k)] ), \label{eq: bound1}
\end{eqnarray}
where it was already defined in equation~(\ref{eq: virtual_arr}) that for all $i$,
\begin{eqnarray*}
&&A_i(k,\QQ(k),\uu(k)) 
=
    \sum_{t = kN}^{kN+N-1} \sum_{s} [\tM^*(t, \tq(t), \QQ(k))]_{is}\cdot \tc_{\tq(t)is}(t) \cdot r_{\tq(t)i}.
\end{eqnarray*}
and we further define
\begin{eqnarray*}
\lambda_i(k,\QQ(k)) &\triangleq &E[A_i(k,\QQ(k),\uu(k))|\QQ(k)]
= N\sum_q \nu_q \sum_{s} [\tM^*(q,t,\QQ(k))]_{is}c_{qis}r_{qi}.
\end{eqnarray*}
Since each client can at most get one webpage slot for each query, we can further bound
$$
\sum_i \!\!A_i^2(k,\QQ,\uu(k)) \!\le \!(N(N-1)L^2+N\!L) (\arg\max_{q,i,s}\{c_{qis} r_{qi}\})^2.$$
Besides,
\begin{eqnarray*}
E[b_i^2(k)] &=& \lceil b_i \rceil^2 \varrho_i + \lfloor b_i \rfloor^2 (1-\varrho_i)
= \lceil b_i \rceil^2 (b_i- \lfloor b_i \rfloor) + \lfloor b_i \rfloor^2 (1-b_i+\lfloor b_i \rfloor).
\end{eqnarray*}
Thus, by defining
\begin{equation}
B_1 \triangleq \frac{1}{2}\Big(
(N(N-1)L^2+N\!L) (\arg\max_{q,i,s}\{c_{qis} r_{qi}\})^2+ \sum_i \lceil b_i \rceil^2 (b_i- \lfloor b_i \rfloor)
+ \lfloor b_i \rfloor^2 (1-b_i+\lfloor b_i \rfloor)
\Big),\nonumber
\end{equation}
%
%
and continuing from inequality~(\ref{eq: bound1}), we have
\begin{eqnarray}
&&E[V(\QQ(k+1))|\QQ(k)=\QQ] - V(\QQ) \nonumber\\
&\le& N\sum_q\nu_q \sum_{i,s} Q_i [\tM^*(q,t,\QQ)]_{is}c_{qis}r_{qi}  
- \sum_i Q_i b_i + B_1  \nonumber\\
&=& -N\sum_q\nu_q \sum_{i,s} \left(\frac{1}{\epsilon} - Q_i\right)[\tM^*(q,t,\QQ)]_{is}c_{qis}r_{qi} \nonumber\\
&&
+ \frac{N}{\epsilon} \sum_q \nu_q \sum_{i,s}
         [\tM^*(q,t,\QQ)]_{is}c_{qis}r_{qi}  
         + B_1 - \sum_i Q_i b_i   \nonumber\\
&=& -N\sum_q\nu_q \sum_{i,s} \left(\frac{1}{\epsilon} - Q_i\right)[\tM^*(q,t,\QQ)]_{is}c_{qis}r_{qi} \nonumber\\
&&
+ \frac{N}{\epsilon} \bar{R}(\tpp^*(k,\QQ)) + B_1 - \sum_i Q_i b_i \label{eq: bound2}\\
&\stackrel{(a)}{\le}& - N\sum_q\nu_q \sum_{i,s} \left(\frac{1}{\epsilon} - Q_i\right)\sum_{M\in \cM_q} \pqM^* M_{is}c_{qis}r_{qi}
+ \frac{N}{\epsilon} \bar{R}(\tpp^*(k,\QQ)) + B_1 - \sum_i Q_i b_i   \nonumber\\
&=& - \frac{N}{\epsilon} \left(\bar{R}(\pp^*) - \bar{R}(\tpp^*(k,\QQ)) \right) + B_1\nonumber\\&&- \sum_i Q_i \cdot     
\left(b_i - N\sum_q \nu_q \sum_{M\in\cM_q} \pqM^* \sum_{s} M_{is}c_{qis}r_{qi}\right), \label{eq: bound3}
\end{eqnarray}
where inequality (a) holds because equation (\ref{eq: dyn_alg_choose_mat}) in the online algorithm is equivalent to 
\begin{eqnarray}
&\forall q,&\tpp_q^*(k,\QQ(k)) 
\in \arg \max_{\{\pqM,\atop M\in \cM_q\}}
    \sum_{M\in\cM_q} \pqM \sum_{i,s} M_{is}c_{qis}r_{qi}\left(\frac{1}{\epsilon}-Q_i(k)\right),
 \nonumber\\   
\label{eq: dyn_max_prob_form}
\end{eqnarray}
which means that evaluating the objective function in (\ref{eq: dyn_max_prob_form}) with $\pp = \pp^*$ cannot achieve a larger value.
Letting
\begin{equation}
B_2\triangleq \min_{i} \{b_i - N\sum_q \nu_q \sum_{M\in\cM_q} \pqM^* \sum_{s} M_{is}c_{qis}r_{qi}\}, \nonumber
\end{equation}
from inequality~(\ref{eq: bound3}), we finally obtain
\begin{eqnarray*}
&&E[V(\QQ(k+1))|\QQ(k)=\QQ] - V(\QQ)
\le - \frac{N}{\epsilon} \left(\bar{R}(\pp^*) - \bar{R}(\tpp^*(k,\QQ)) \right) + B_1
    - B_2\sum_i Q_i.
\end{eqnarray*}

\subsection{Proof of Theorem~\ref{thm: revenue_converge}}
\label{sec: proof_revenue_converge}
The first inequality which shows that the online algorithm cannot do better than the offline optimal solution is too obvious, so we just ignore it here (proving it in a very rigorous way is also very easy, after defining the ``per-client revenue region'' in Subsection \ref{sec: multi_que_LB} and then using the fact that the average revenue vector $\blambda$ corresponding to our online algorithm falls inside that region,  according to inequality \eqref{eq: long-run spending} which is implied by stability). 
 
We now focus on the second inequality, i.e., the $O(\epsilon)$ convergence bound. From Lemma~\ref{lemma: bounded_drift2},
\begin{eqnarray*}
&&E\left[ \bar{R}(\pp^*) - \bar{R}(\tpp^*(k,\QQ(k))) \right]\\
&\le& \frac{\epsilon}{N}\cdot E\Bigg[B_1 - B_2\sum_i \QQ(k) + V(\QQ(k)) - E[V(\QQ(k+1))|\QQ(k)]\Bigg]\\
&\le& \frac{\epsilon}{N} \cdot (B_1 - E[V(\QQ(k))] - E[V(\QQ(k+1))]),
\end{eqnarray*}
Adding the terms for $0\le k \le K-1$ and dividing by $K$, we get
\begin{eqnarray*}
\frac{1}{K} \sum_{k=0}^{K-1} E\left[ \bar{R}(\pp^*) - \bar{R}(\tpp^*(k,\QQ(k))) \right]
&\le& \frac{\epsilon}{N} \left(B_1 - \frac{E[V(\QQ(K))]}{K} + \frac{V(\QQ(0))}{K} \right)\\
&\le&\frac{\epsilon}{N} \left(B_1 + \frac{V(\QQ(0))}{K} \right).
\end{eqnarray*}
Since $V(\QQ(0)) < \infty$, we get the following limit expression:
\begin{equation}
\lim_{K\rightarrow \infty} \frac{1}{K} \sum_{k=0}^{K-1} E\left[ \bar{R}(\pp^*) - \bar{R}(\tpp^*(k,\QQ(k))) \right]
\le \frac{ B_1\epsilon}{N}. \label{eq: rev_diff_ub}
\end{equation}
Finally, because
\begin{eqnarray*}
&&E\left[ N \bar{R}(\pp^*) - R(k) \right]
= E\left[E\left[ N\bar{R}(\pp^*) - R(k) | \QQ(k) \right]\right]
= N\cdot E\left[ \bar{R}(\pp^*) - \bar{R}(\tpp^*(k,\QQ(k))) \right],
\end{eqnarray*}
inequality~(\ref{eq: rev_diff_ub}) is equivalent to
$$
\lim_{K\rightarrow \infty} E\left[ \bar{R}(\pp^*) - \frac{1}{KN}\sum_{k=0}^{K-1} R(k) \right] \le \frac{B_1\epsilon}{N}.
$$

\subsection{Proof of Corollary~\ref{cor: revenue_converge_estimate}}
\label{sec: proof_revenue_converge_estimate}
Continuing from inequality~(\ref{eq: bound2}) in Appendix \ref{sec: proof_bounded_drift2} (the proof of Lemma~\ref{lemma: bounded_drift2}), we get
\begin{eqnarray}
&&E[V(\QQ(k+1))|\QQ(k)=\QQ] - V(\QQ) \nonumber\\
&\stackrel{(a)}{\le}&\!\!\!-\frac{1}{1+\Delta}N \sum_q\nu_q \sum_{i,s} \left(\frac{1}{\epsilon} - Q_i\right)[\tM^*(q,t,\QQ)]_{is}\hc_{qis}r_{qi} 
+ \frac{N}{\epsilon} \bar{R}(\tpp^*(k,\QQ)) + B_1 - \sum_i Q_i b_i \nonumber \\
&\stackrel{(b)}{\le}&\!\!\! -\frac{1}{1+\Delta} N\sum_q\nu_q \sum_{i,s} \left(\frac{1}{\epsilon} - Q_i\right)\sum_{M\in \cM_q} \pqM^* M_{is} \hc_{qis}r_{qi}
 + \frac{N}{\epsilon} \bar{R}(\tpp^*(k,\QQ)) + B_1 - \sum_i Q_i b_i   \nonumber
\end{eqnarray}
\begin{eqnarray}
&\stackrel{(c)}{\le}&\!\!\! -\frac{1-\Delta}{1+\Delta} N\sum_q\nu_q \sum_{i,s} \left(\frac{1}{\epsilon} - Q_i\right)\sum_{M\in \cM_q} \pqM^* M_{is} c_{qis}r_{qi}
+ \frac{N}{\epsilon} \bar{R}(\tpp^*(k,\QQ)) + B_1 - \sum_i Q_i b_i   \nonumber\\
&=&\!\!\! - \frac{N}{\epsilon} \left(\frac{1-\Delta}{1+\Delta}\bar{R}(\pp^*) - \bar{R}(\tpp^*(k,\QQ)) \right) + B_1 \nonumber\\
&&  
- \sum_i Q_i \cdot \left(b_i - \frac{1-\Delta}{1+\Delta} N \sum_q \nu_q \sum_{M\in\cM_q} \pqM^* \sum_{s} M_{is}c_{qis}r_{qi}\right).
\label{eq: bound3_ctr_est}
\end{eqnarray}
Here, inequalities (a) and (c) hold respectively because $\hcc \le \cc (1+\Delta)$ and $\hcc \ge \cc (1-\Delta)$, with the fact that all the coefficients in this summation are nonnegative. Inequality (b) holds because equation (\ref{eq: dyn_alg_choose_mat_estimate}) in the online algorithm with estimated click-through rates is equivalent to
\begin{eqnarray}
&\forall q,&\tpp_q^*(k,\QQ(k)) 
\in \arg \max_{\{\pqM,\atop M\in \cM_q\}}
    \sum_{M\in\cM_q} \pqM \sum_{i,s} M_{is}\hc_{qis}r_{qi}\left(\frac{1}{\epsilon}-Q_i(k)\right),
    \label{eq: dyn_max_prob_form2}
\end{eqnarray}
which means that evaluating the objective function in (\ref{eq: dyn_max_prob_form2}) with $\pp = \pp^*$ cannot achieve a larger value. Letting
\begin{equation}
B'_2\triangleq \!\min_{i} \left\{b_i  \!- \! \frac{1-\Delta}{1+\Delta} \cdot N\sum_q \nu_q  \! \!\sum_{M\in\cM_q} \! \! \pqM^* \sum_{s} M_{is}c_{qis}r_{qi}\right\}, \nonumber
\end{equation}
from inequality~(\ref{eq: bound3}), we finally obtain
\begin{equation*}
E[V(\QQ(k+1))|\QQ(k)=\QQ] - V(\QQ)
\le  - \frac{N}{\epsilon} \left(\frac{1-\Delta}{1+\Delta} \bar{R}(\pp^*) - \bar{R}(\tpp^*(k,\QQ)) \right) + B_1 - B'_2\sum_i Q_i.
\end{equation*}
Therefore, similarly as in the proof of Theorem~\ref{thm: revenue_converge}, we can finally show that
\begin{eqnarray*}
\lim_{K\rightarrow \infty} E\left[\frac{1}{KN}\sum_{k=0}^{K-1} R(k) \right]
&\ge&
\left(\frac{1-\Delta}{1+\Delta}\right)\cdot \bar{R}(\pp^*) - \frac{B_1\epsilon}{N}.
\end{eqnarray*}

\

\subsection{Proof of Lemma~\ref{lemma: bounded_drift_model2}}
\label{sec: proof_bounded_drift_model2}
By a similar approach as in the proof of Lemma \ref{lemma: bounded_drift2} (Appendix \ref{sec: proof_bounded_drift2}), we have \begin{eqnarray}
&&E[V(\QQ(k+1))|\QQ(k)=\QQ] - V(\QQ) \nonumber\\
&\le&\!\!\! -N\sum_q\nu_q \sum_{i,s} \left(Q_i + \frac{c_{qis}}{\epsilon}\right)[\tM^*(q,t,\QQ)]_{is} 
+ \frac{N}{\epsilon} \bar{J}(\tpp^*(k,\QQ)) + D_1 + \sum_i Q_i m_i \nonumber\\
&\stackrel{(a)}{\le}&\!\!\! - N\sum_q\nu_q \sum_{i,s} \left(Q_i + \frac{c_{qis}}{\epsilon}\right)\sum_{M\in \cM_q} \hpqM M_{is}
+ \frac{N}{\epsilon} \bar{J}(\tpp^*(k,\QQ)) + D_1 + \sum_i Q_i m_i   \nonumber\\
&=&\!\!\! - \frac{N}{\epsilon} \left(\bar{J}(\hpp) - \bar{J}(\tpp^*(k,\QQ)) \right) + D_1- \sum_i Q_i
\left(N\sum_q \nu_q \sum_{M\in\cM_q} \hpqM \sum_{s} M_{is} - m_i\right), \label{eq: bound3_model2}
\end{eqnarray}
where $D_1$ is an upper bound on $\frac{1}{2} \sum_i ( E[S_i^2(k,\QQ,\uu(k))] + E[\tm_i^2(k)])$ and defined as
\begin{eqnarray*}
D_1 \!\!\!&\triangleq&\!\!\! \frac{1}{2} \Big(N(N-1) L^2+NL
 + \sum_i \lceil m_i \rceil^2 (m_i- \lfloor m_i \rfloor) 
 + \lfloor m_i \rfloor^2 (1-m_i+\lfloor m_i \rfloor)
\Big).
\end{eqnarray*}
Note that inequality \eqref{eq: bound3_model2} has the same form as inequality \eqref{eq: bound3} in the proof of Lemma \ref{lemma: bounded_drift2}, except that the offline optimum $\pp^*$ is replaced by some $\hpp \in \mathcal{F}$. 
Letting
\begin{equation}
D_2\triangleq \min_{i} \{N\sum_q \nu_q \sum_{M\in\cM_q} \hpqM \sum_{s} M_{is}-m_i\}, \nonumber
\end{equation}
it is always possible to pick a $\hpp\in\mathcal{F}$ such that $D_2 >0$  (unless $\mathcal{F}$ is a degenerated set which has at most one element). We further bound the above inequality as  
\begin{eqnarray*}
&&E[V(\QQ(k+1))|\QQ(k)=\QQ] - V(\QQ)
\le \frac{D_3}{\epsilon}  + D_1
    - D_2 \sum_i Q_i.
\end{eqnarray*}
Here, 
$D_3 \triangleq N \cdot \max_{\pp\in \mathcal{F}_0}\avgJ(\pp)$ 
where $\mathcal{F}_0$ is defined in \eqref{eq: feasible_relaxed}. 
This concludes our proof.


\subsection{Short-Term Clients and Non-Stationary Query Arrivals}
\label{sec: mixed}
We focus on the click-through rate maximization problem, although a similar model and solution can be used for revenue maximization problem. 

First, consider how to include short-term clients in the system. Let us index long-term clients from $1$ to $n$, the $i^{\rm th}$ of which has an average impression requirement of $m_i$ per requirement cycle. There are further $\tn$ types of short-term clients indexed by $n+1 \le i \le n+\tn$. Each short-term client of type $i$ has a impression requirement of $l_i$ per contract term. Without loss of generality, we assume that the contract term of any short-term client is equal to one requirement cycle. In each requirement cycle $k$, there are $X_i(k)$ clients of type $i$ in the system, where $X_i(k)$ follows a stationary stochastic process with mean $x_i$ and $X_i(k)$ is known at the beginning of requirement cycle $k$. 


Correspondingly in an ad assignment matrix $M$, the first $n$ rows and the subsequent $\tn$ rows represent the $n$ long-term clients and the $\tn$ types of short-term clients, respectively. If short-term type $j$ is assigned to some webpage slot, one out of $X_j(k)$ clients of this type is chosen uniformly at random due to their homogeneity. 

Additionally, for a short-term client of type $i$, the algorithm is actually aimed to satisfy at least only $(1-\alpha_i) l_i$, where $\alpha_i \in [0,1]$ is called ``unfulfilled rate'' for clients of type $i$ and to be determined by the algorithm. A strictly convex and monotonically increasing function $\phi(\alpha_i) \in [0,\infty)$ is then introduced to measure the ``unhappiness'' of short-term clients about unfulfilled impression requirements, and deducted from the original objective function ``overall average click-through rate'' in \eqref{obj: max_click-rate} after scaled by some predetermined weight $w_i$ which reflects the importance of the new metric ``unfulfilled rate.''  

The second extension from the original model is to consider a more general query arrival pattern. We introduce a new time scale ``stationary-arrival period'' between the fast one ``time slot'' $t$ and the slow one ``requirement cycle'' $k$, namely one requirement cycle equals $H$ stationary-arrival periods (assuming that $N/H \in \mathcal{Z}^+$ and usually $N/H \gg 1$), and we assume that query arrivals with respect to each keyword $q$ form a stationary stochastic process with rate $\nu_q(h)$ within the $h^{\rm th}$ stationary-arrival period in one requirement cycle for all $1\le h \le H$. This is a more reasonable assumption for the query arrival pattern in the real Internet. For example, in one day, the query arrivals are stationary within each individual hour, non-stationary across different hours, and stationary in the same hour across different days. This corresponds to $H=24$, although setting a contract term (already assumed to be equal to one requirement cycle) as one day would only be a simplification for ease of exposition. Based on this example, in the following text we are going to use ``day'' and ``hour'' instead of ``requirement cycle'' and ``stationary-arrival period'' to better describe the basic ideas.   

In summary, the new optimization problem is formulated as
\begin{eqnarray*}
\max_{\{\pp(h),\forall h;~\bm{\alpha}\}}\!\!\!\!\!\!\!\!\!&& \frac{1}{H} \sum_q\sum_{h=1}^H \nu_q(h) \!\!\!\! \sum_{M\in \cM_q}\!\!\pqM(h)\!\!\!\!\!\! \sum_{1\le i \le n+\tn,s} \!\!\!\!\!\!\!\!M_{is} c_{qis} 
- \sum_{i=n+1}^{n+\tn} w_i \phi(\alpha_i)\nonumber
\end{eqnarray*}
subject to
\begin{eqnarray}
&&\frac{N}{H} \sum_q \sum_{h=1}^H \nu_q(h) \!\!\!\! \sum_{M\in \cM_q}\!\!\pqM(h) \sum_{s} M_{is}
\ge 
\left\{
\begin{array}{c@{,~~\forall~}c}
m_i & 1\le i \le n\\
(1-\alpha_i) l_i x_i& n+1 \le i \le n+\tn  
\end{array}
\right.
\nonumber
\end{eqnarray}
and
\begin{equation*}
0\leq \pqM(h) \leq 1,~\forall q, ~M \in \cM_q, ~1\le h \le H;\quad
\sum_{M\in \cM_q} \!\!\!\!\pqM(h) \leq 1,~\forall q, ~1\le h \le H.
\end{equation*}

The only modification in the online algorithm described in Subsection \label{sec: alg_perf_model2} is to add the following two steps specially for each type of short-term clients:
\begin{itemize}
\item At the beginning of the $k^{\rm th}$ day, update 
$$
\alpha^*_i(k) =\psi \left(\frac{ l_i X_i(k) \cdot Q_i(k)}{H w_i}\right),
$$
which corresponds to the target ``unfulfilled rate'' for each type of short-term clients in this day. Here, the function $\psi \triangleq \left[\frac{d\phi}{d\alpha}\right]^{-1}$.

\item At the end of the $k^{\rm th}$ day, ``credit queue'' $i$ maintained for type $i$ of short-term clients is updated as
\begin{eqnarray*}
Q_i(k+1) &=& Q_i(k) + (1-\alpha_i^*(k)) \cdot l_i X_i(k)
- S_i(k,\QQ(k),\qq(k)),
\end{eqnarray*}
where $S_i(k,\QQ(k),\qq(k))$ is defined in \eqref{eq: virtual_depart}.
\end{itemize} 

The conclusions and proofs about near-optimality of the objective value, queueing stability and upper bound on the expected queue length are similar as those shown for the original problem in Subsection \ref{sec: perf_model2} and hence omitted here. 

Note that the online algorithm is still ``oblivious'' to the query arrivals, when the arrival processes become non-stationary to some extent. This is an artifact of dual decomposition w.r.t. each hour $h$, in addition to a decomposition w.r.t. each keyword $q$ as we have seen before.


\end{document}